\documentclass[a4paper,onecolumn,11pt,unpublished]{quantumarticle}
\pdfoutput=1
\usepackage[utf8]{inputenc}
\usepackage[english]{babel}
\usepackage[T1]{fontenc}
\usepackage{amsmath,amssymb,amsfonts,amsthm}
\usepackage{bm}  % Bold symbols, e.g. greek letters
\usepackage{caption}
\usepackage{subcaption}
\usepackage{qcircuit}
\usepackage{hyperref}
\usepackage{xcolor}
\usepackage[square,numbers,sort]{natbib}
\definecolor{fraunhofergreen}{RGB}{23,156,125} 
\hypersetup{
    colorlinks=fraunhofergreen,
    linkcolor=fraunhofergreen,
    urlcolor=fraunhofergreen,
    filecolor=fraunhofergreen,
    citecolor=fraunhofergreen,
}

\usepackage{tikz}

\newcommand\numberthis{\addtocounter{equation}{1}\tag{\theequation}}

\newtheorem{theorem}{Theorem}
\newtheorem{lemma}{Lemma}

\makeatletter
\newcommand{\doublewidetilde}[1]{{%
  \mathpalette\double@widetilde{#1}%
}}
\newcommand{\double@widetilde}[2]{%
  \sbox\z@{$\m@th#1\widetilde{#2}$}%
  \ht\z@=.85\ht\z@
  \widetilde{\box\z@}%
}
\makeatother

\renewcommand{\i}{\mathrm{i}}

\newcommand{\C}{\mathbb{C}}
\newcommand{\N}{\mathbb{N}}
\newcommand{\R}{\mathbb{R}}

\newcommand{\numDim}{D}

\newcommand{\domain}{\Omega}
\newcommand{\laplace}{L}
\newcommand{\laplaceOne}{\laplace_{1}}
\newcommand{\laplaceOneh}{\laplace_{1,h}}
\newcommand{\laplaceOnehScaled}{\widetilde{\laplace}_{1,h}}

\newcommand{\laplaceTwohScaled}{\widetilde{\laplace}_{2,h}}

\newcommand{\laplaceThreehScaled}{\widetilde{\laplace}_{3,h}}

\newcommand{\laplaceD}{\laplace_{\numDim}}
\newcommand{\laplaceDh}{\laplace_{\numDim, h}}
\newcommand{\laplaceDhScaled}{\widetilde{\laplace}_{\numDim, h}}
\newcommand{\numDofOne}{N}
\newcommand{\DOneh}{D_{1,h}}
\newcommand{\DOnehScaled}{\widetilde{D}_{1,h}}
\newcommand{\DOnehScaledTwo}{\doublewidetilde{D}_{1,h}}
\newcommand{\DOnehScaledx}{\widetilde{D}_{1,h}^{(0)}}
\newcommand{\DOnehScaledy}{\widetilde{D}_{1,h}^{(1)}}
\newcommand{\SOneh}{Q_{1,h}}
\newcommand{\SOnehScaled}{\widetilde{Q}_{1,h}}
\newcommand{\numDofQC}{N}
\newcommand{\numDof}{N_{\numDim}}
\newcommand{\numSytemBits}{n}
\newcommand{\numAncillaDof}{M}
\newcommand{\numAncillaBits}{m}
\newcommand{\numAncillaKDof}{\widehat{D}}
\newcommand{\numAncillaKBits}{\widehat{d}}

\newcommand{\vecv}{\mathbf{v}}
\newcommand{\vecw}{\mathbf{w}}

\newcommand{\vecDomain}{\bm{\domain}}

\newcommand{\vecx}{\mathbf{x}}
\newcommand{\jind}[1]{j^{(#1)}}

\newcommand{\Shift}[1]{S^{(#1)}}
\newcommand{\ShiftPlusMinus}{S_{\pm}}
\newcommand{\ShiftMinus}{S_{\boldsymbol{-}}}
\newcommand{\ShiftPlus}{S_{\boldsymbol{+}}}

\newcommand{\scalingFactor}{\alpha_{\numDim}}
\newcommand{\scalingFactorBE}{\alpha}

\newcommand{\vecxentry}[1]{x^{(#1)}}

\newcommand{\lambdaMax}[1]{\lambda_{#1, \mathrm{max}}}

\newcommand{\bit}[1]{\texttt{#1}}
\newcommand{\ketSimple}[1]{| #1 \rangle}
\newcommand{\ket}[1]{\left| #1 \right\rangle}
\newcommand{\ketBit}[1]{\left| \bit{#1} \right\rangle}
\newcommand{\bra}[1]{\left\langle #1 \right|}

\newcommand{\norm}[1]{\|#1\|}
\newcommand{\pmat}[1]{\begin{pmatrix} #1 \end{pmatrix}}

\newcommand{\gridPlotWidth}{4.5cm}%
\newcommand{\heatmapPlotWidth}{6.0cm}%

\begin{document}

\title{Efficient and Explicit Block Encoding of Finite Difference Discretizations of the Laplacian}

\author{Andreas Sturm}
\affiliation{
    Fraunhofer-Institut für Arbeitswirtschaft und Organisation IAO,
	Nobelstra{\ss}e 12,
	70569 Stuttgart,
	Germany}
\email{andreas.sturm@iao.fraunhofer.de}
\author{Niclas Schillo}
\affiliation{
    Fraunhofer-Institut für Arbeitswirtschaft und Organisation IAO,
	Nobelstra{\ss}e 12,
	70569 Stuttgart,
	Germany}
\affiliation{
    Universit\"at Stuttgart Institut für Arbeitswissenschaft und Technologiemanagement IAT,
    Nobelstra{\ss}e 12,
    70569 Stuttgart,
    Germany}
\email{niclas.schillo@iao.fraunhofer.de}
\maketitle

\section*{Abstract}
The data input model is a fundamental component of every quantum algorithm, as its efficiency is crucial for achieving potential speed-ups over classical methods. For quantum linear algebra tasks that utilize quantum eigenvalue or singular value transformations, block encoding is the established technique for accessing matrix data. A key application of this is solving partial differential equations, where the Laplacian operator and its finite difference discretization serve as foundational examples. In this paper, we present an efficient and explicit block encoding method that enhances existing approaches in key aspects. We detail the construction of the quantum algorithm and illustrate how it leverages the unique structure of finite difference discretizations. Furthermore, we analytically derive the scaling of the sub-normalization factor and of the success probability of the block encoding with respect to the problem dimension, the grid width of the finite difference grid and the regularity of the exact solution, and we give resource estimates.

\section{Introduction}
Partial differential equations (PDEs) are fundamental in modeling a wide array of phenomena across various scientific disciplines such as physics, engineering, and finance. Among the various operators used in the formulation of PDEs, the Laplacian operator holds a prominent position as it appears in key equations, including the Poisson, heat, and wave equation \cite{Evans.2022}. These PDEs serve as foundational examples of elliptic, parabolic, and hyperbolic differential equations. Their study is not only the basis for understanding more complex equations, but also the development of numerical methods is usually based on these foundational equations.

The most prominent numerical scheme to approximate the Laplacian is the finite difference method \cite{Smith.1993,Thomas.1995}. This method discretizes the continuous domain by estimating derivatives using values at discrete grid points, resulting in a discrete representation of the Laplacian as a sparse, banded matrix. Numerical linear algebra techniques are then employed to approximate the exact solution of the PDE, typically involving matrix functions calculated as polynomials or rational functions of the discretized Laplacian \cite{Hairer.1991,Sogabe.2022}.  However, the size of the discretized Laplacian scales with  $\mathcal{O}(\numDof)$, where $\numDof = \numDofOne^\numDim$ depends on the number of grid points $\numDofOne$ in each spatial dimension  and the number of dimensions $\numDim$. This growth ultimately renders any scheme intractable for classical computers.

The emergence of quantum computing offers new possibilities for addressing these computational challenges. Quantum computers can store a vector of dimension $\numDof$ in the amplitudes of just 
$\log\numDof$ qubits, potentially alleviating the memory limitations faced by classical systems. However, to leverage this advantage, the discrete Laplacian must be efficiently encoded onto a quantum computer. Given that quantum computing relies solely on unitary transformations, this encoding involves embedding the matrix $A$ into a larger unitary matrix, a process known as block encoding \cite{Low.2019,Gilyen.2019}. Previous works such as \cite{Berry.2007,Berry.2015,Childs.2017,Low.2019,Gilyen.2019} have explored the use of block encoding of specific sparse matrices within larger algorithms. But the block encoding there is based on oracles that provide black box access to the non-zero pattern and the entries of $A$. Constructing explicit quantum circuits to implement these oracles is a complex task, and their exact formulations can significantly increase the quantum resources required.

For the discrete Laplacian with periodic boundary conditions, explicit circuits have been constructed in prior research. In the one dimensional case \cite{Camps.2022} achieved this for the oracles of \cite{Gilyen.2019} and for arbitrary dimensions \cite{Kharazi.2025} developed circuits using the method of linear combinations of unitaries. However, in these approaches, the matrix $A$ is scaled to below the unity norm, which adversely affects both the success probability of the block encoding and the quantum algorithms that utilize it. For example, this scaling increases the simulation time of methods targeting Hamiltonian simulation problems, which calculate the time evolution $e^{t A} \ket{v_0}$ of an initial state $\ket{v_0}$.
% multiplies the simulation time if the block encoding is used for Hamiltonian simulation type problems.

In this paper, we present a block encoding for the discrete Laplacian with periodic boundary conditions in arbitrary dimensions $\numDim$, that achieves the optimal sub-normalization factor when $\numDim$ is a power of two. Our method offers several additional benefits: it is fully explicit, allowing for easy translation into the native gate set of quantum computers; it provides an exact encoding of $A$; and it is resource-efficient, requiring only
$\mathcal{O}(\log \numDof)$ $T$ gates. In conjunction with quantum signal processing \cite{Low.2017,Low.2019} or more generally quantum singular value transformation \cite{Gilyen.2019,Martyn.2021} that facilitate a broad range of matrix function implementations, our method shows a path for solving PDEs on quantum computers. 

The paper is structured as follows: Section~\ref{sec:preliminaries} introduces our notation and the concept of block encoding. In Section~\ref{sec:finite-differences} we discuss the finite difference discretization of the Laplacian. Our main results are given in Sections~\ref{sec:block-encoding-laplace-1d} and \ref{sec:block-encoding-laplace}, where the first is dedicated to the block encoding of the discrete Laplacian in one dimension and the latter generalizes this to arbitrary dimensions. Section~\ref{sec:numerical-experiments} concludes the paper with numerical experiments. In the appendix we briefly outline how our ideas can be adapted to first-order differential operators. 

\section{Preliminaries}
\label{sec:preliminaries}
We start by briefly introducing the required notations in quantum computing and the concept of block encoding. For a more detailed introduction, we refer to the textbooks \cite{Nielsen.2012,Rieffel.2014} for the former and \cite{Lin.2022} for the latter.
\subsection{Notations and Conventions}
\label{sec:notations}
In this paper we follow the standard conventions in the quantum computing literature where the Dirac notations $\bra{\cdot}$ and $\ket{\cdot}$ are used to denote row and column vectors, respectively. In particular, $\ketBit{0}$ and $\ketBit{1}$ are used to denote respectively the two canonical basis vectors $(1, 0)^T$ and $(0, 1)^T$ of $\C^2$. For a bitstring $\bit{b} = \bit{b}_0 \bit{b}_1 \dots \bit{b}_{\numSytemBits-1}$ of $\numSytemBits$ bits $\bit{b}_i \in \{0, 1\}$, the notation $\ket{\bit{b}}$ means the tensor product of the respective $\ket{\bit{b}_i}$, i.e.
\begin{equation*}
    \ket{\bit{b}} % _n 
    = \ket{\bit{b}_0 \bit{b}_1 \dots \bit{b}_{\numSytemBits-1}}
    = \ket{\bit{b}_0} \ket{\bit{b}_1} \dots \ket{\bit{b}_{\numSytemBits-1}}
    = \ket{\bit{b}_0} \otimes \ket{\bit{b}_1} \otimes \dots \otimes \ket{\bit{b}_{\numSytemBits-1}}\ .
\end{equation*}
Here, we denote by $\otimes$ the Kronecker product. 

For a non-negative integer $j \in \{0, 1, \dots \numDofQC-1\}$ with $\numDofQC = 2^\numSytemBits$, the binary representation in the big endian convention is given by
\begin{equation*}
    j 
    = [ \bit{j}_0 \bit{j}_1 \dots \bit{j}_{\numSytemBits-1} ]
    = \bit{j}_0 2^{\numSytemBits-1} + \bit{j}_1 2^{\numSytemBits-2} + \dots
    + \bit{j}_{\numSytemBits-1} 2^0 \ .
\end{equation*}
Accordingly, we use $\ket{j} = \ket{\bit{j}_0 \bit{j}_1 \dots \bit{j}_{\numSytemBits-1}}$. Note that with this convention $\ket{j}$, $j=0, 1, \dots,  \numDofQC$, are the canonical basis vectors of $\C^\numDofQC$.

Any quantum state $\ket{v}$ of $\numSytemBits$ qubits can be expressed by a vector in $\C^\numDofQC$ with unit norm. It can be written as
\begin{equation*}
    \ket{v}
    = \sum_{j = 0}^{\numDofQC - 1}
    v_j \ket{j} \ ,
    % \qquad
    % \sum_{\ell = 0}^{\numDof - 1} |v_\ell|^2 
    % = 1 \ .
\end{equation*}
where $v_j$ are the amplitudes of $v$ which satisfy $\sum_{j = 0}^{\numDofQC - 1} |v_j|^2 = 1$. Valid operations on quantum states are given by unitary matrices, which are called gates in the quantum computing context. Qubits and sequences of quantum gates are usually visualized as quantum circuits, where wires represent the qubits and boxes the gates. In this paper, we number the wires from bottom to top and we abbreviate multi-wires with a slash. For example, for $\ket{j} = \ket{\bit{j}_0 \bit{j}_1 \bit{j}_2}$ we have
\begin{equation*}
    \begin{array}{c}
    \Qcircuit @C=1em @R=.4em {
        & 
        &
        &
        & \push{\ket{\bit{j}_2} \ }
        & \qw
        & \qw
        \\
        \lstick{\ket{j}} 
        & {/} \qw 
        & \qw
        & \push{=}
        & \push{\ket{\bit{j}_1} \ }
        & \qw
        & \qw
        \\
        &
        &
        &
        & \push{\ket{\bit{j}_0} \ }
        & \qw
        & \qw
    } 
    \end{array}
    \ .
\end{equation*}
Important gates are the Pauli-$X$, -$Y$, and -$Z$ gates and the Hadamard gate $H$, which are given by
\begin{equation*}
    X = \pmat{0 & 1 \\ 1 & 0} \ ,
    \quad
    Y = \pmat{0 & -\i \\ \i & 0} \ ,
    \quad
    Z = \pmat{1 & 0 \\ 0 & -1} \ , 
    \quad
    H = \tfrac{1}{\sqrt2} \pmat{1 & 1 \\1 & -1} \ .
\end{equation*}
The $Y$-rotation about an angle $\theta$ is defined as $RY(\theta) = \exp(-\i \tfrac{\theta}{2} Y)$. Quantum gates are read from left to right, e.g. for two single qubit systems $\ket{v}$ and $\ket{w}$ we have
\begin{equation*}
    \begin{array}{c}
    \Qcircuit @C=1em @R=.4em {
    \lstick{\ket{v}}
    & \gate{X}
    & \gate{Z}
    & \qw
    \\
    \lstick{\ket{w}}
    & \gate{RY(\theta)}
    & \qw
    & \qw
    }
    \end{array}
    \quad = \quad
    RY(\theta) \ket{w} Z X \ket{v} \ .
\end{equation*}
For a gate $U$ the $\ketBit{0}$-controlled version of it is given by
\begin{equation*}
    \begin{array}{c}
    \Qcircuit @C=1em @R=.7em {
        \lstick{\ket{j}} 
        & {/} \qw 
        & \gate{U}
        & \qw
        \\
        \lstick{\ketBit{c}}
        & \qw
        & \ctrlo{-1}
        & \qw
    }   
    \end{array}
    = \ketBit{0}\bra{\bit{0}} \otimes U 
    + \ketBit{1}\bra{\bit{1}} \otimes I \ ,
    \qquad
    \ketBit{c}\ket{j} \to 
    \left\{
    \begin{array}{ll}
         \ketBit{0} U \ket{j} \, , 
         & \bit{c} = 0 \ ,  
         \\
         \ketBit{1} \ket{j} \, ,
         & \bit{c} = 1 \ ,
    \end{array}
    \right.
\end{equation*}
and the $\ketBit{1}$-controlled version by
\begin{equation*}
    \begin{array}{c}
    \Qcircuit @C=1em @R=.7em {
        \lstick{\ket{j}} 
        & {/} \qw 
        & \gate{U}
        & \qw
        \\
        \lstick{\ketBit{c}}
        & \qw
        & \ctrl{-1}
        & \qw
    }
    \end{array}
    = \ketBit{1}\bra{\bit{1}} \otimes U 
    + \ketBit{0}\bra{\bit{0}} \otimes I \ ,
    \qquad
    \ketBit{c}\ket{j} \to 
    \left\{
    \begin{array}{ll}
         \ketBit{0} \ket{j} \, ,
         & \bit{c} = 0 \ ,  
         \\
         \ketBit{1} U \ket{j} \, ,
         & \bit{c} = 1 \ .
    \end{array}
    \right.
\end{equation*}
Multi-controlled gates are constructed similarly. For bundled wires, we give the control as an integer number, e.g. for two control qubits $\ket{c} = \ket{\bit{c}_0 \bit{c}_1}$ we have
\begin{equation*}
    \begin{array}{c}
\Qcircuit @C=1em @R=.7em {
    \lstick{\ket{j}} 
    & {/} \qw 
    & \gate{U}
    & \qw
    \\
    \lstick{\ket{c}}
    & {/} \qw
    & \measure{0} \qwx[-1] %\ctrl{-1}
    & \qw
}
\end{array}
\begin{array}{c}
=
\end{array}
\qquad \ \
\begin{array}{c}
\Qcircuit @C=1em @R=.8em {
    \lstick{\ket{j}} 
    & {/} \qw 
    & \gate{U}
    & \qw
    \\
    \lstick{\ket{\bit{c}_1}}
    & \qw
    & \ctrlo{-1}
    & \qw
    \\
    \lstick{\ket{\bit{c}_0}}
    & \qw
    & \ctrlo{-1}
    & \qw
}
\end{array}
\ ,
\qquad \! \qquad
\begin{array}{c}
\Qcircuit @C=1em @R=.7em {
    \lstick{\ket{j}} 
    & {/} \qw 
    & \gate{U}
    & \qw
    \\
    \lstick{\ket{c}}
    & {/} \qw
    & \measure{1} \qwx[-1] %\ctrl{-1}
    & \qw
}
\end{array}
\begin{array}{c}
= 
\end{array}
\qquad \ \
\begin{array}{c}
\Qcircuit @C=1em @R=.8em {
    \lstick{\ket{j}} 
    & {/} \qw 
    & \gate{U}
    & \qw
    \\
    \lstick{\ket{\bit{c}_1}}
    & \qw
    & \ctrl{-1}
    & \qw
    \\
    \lstick{\ket{\bit{c}_0}}
    & \qw
    & \ctrlo{-1}
    & \qw
}
\end{array}
\end{equation*}
and
\begin{equation*}
    \begin{array}{c}
\Qcircuit @C=1em @R=.7em {
    \lstick{\ket{j}} 
    & {/} \qw 
    & \gate{U}
    & \qw
    \\
    \lstick{\ket{c}}
    & {/} \qw
    & \measure{2} \qwx[-1] %\ctrl{-1}
    & \qw
}
\end{array}
\begin{array}{c}
= 
\end{array}
\qquad \ \
\begin{array}{c}
\Qcircuit @C=1em @R=.8em {
    \lstick{\ket{j}} 
    & {/} \qw 
    & \gate{U}
    & \qw
    \\
    \lstick{\ket{\bit{c}_1}}
    & \qw
    & \ctrlo{-1}
    & \qw
    \\
    \lstick{\ket{\bit{c}_0}}
    & \qw
    & \ctrl{-1}
    & \qw
}
\end{array}
\ ,
\qquad \! \qquad
\begin{array}{c}
\Qcircuit @C=1em @R=.7em {
    \lstick{\ket{j}} 
    & {/} \qw 
    & \gate{U}
    & \qw
    \\
    \lstick{\ket{c}}
    & {/} \qw
    & \measure{3} \qwx[-1] % \ctrl{-1}
    & \qw
}
\end{array}
\begin{array}{c}
= 
\end{array}
\qquad \ \
\begin{array}{c}
\Qcircuit @C=1em @R=.8em {
    \lstick{\ket{j}} 
    & {/} \qw 
    & \gate{U}
    & \qw
    \\
    \lstick{\ket{\bit{c}_1}}
    & \qw
    & \ctrl{-1}
    & \qw
    \\
    \lstick{\ket{\bit{c}_0}}
    & \qw
    & \ctrl{-1}
    & \qw
}
\end{array}
\ .
\end{equation*}

In the subsequence we will need shift-operators defined by 
\begin{equation}
    \Shift{k} \ket{j}
    = \ket{j + k \mod \numDofQC} \ ,
    \qquad
    \text{for }
    k \in \{ 0, \pm 1, \dots, \pm(\numDofQC-1) \} \ .
    \label{eq:shift}
\end{equation}
In fact, we will only need $\Shift{-1}$ and $\Shift{+1}$, which we abbreviate by $\ShiftMinus$ and $\ShiftPlus$, respectively. Moreover, from now on we drop $\mod \numDofQC$ in the kets for a clearer notation. The gates $\ShiftPlusMinus$ can be implemented using only $X$ and (multi-)controlled versions of $X$ \cite[Chapter 6]{Rieffel.2014}. For example, for $\numSytemBits = 3$ qubits these circuits are given as
\begin{equation}
    \begin{array}{c}
    \Qcircuit @C=0.8em @R=0.5em {
% & \lstick{\ket{q_2}}
& \multigate{2}{\ShiftMinus}
& \qw
&
&
& \gate{X}
& \ctrl{1}
& \ctrl{2}
& \qw
\\
% & \lstick{\ket{q_1}}
& \ghost{\ShiftMinus}
& \qw
& \push{=}
&
& \qw
& \gate{X} %\targ
& \ctrl{1}
& \qw
\\
% & \lstick{\ket{q_0}}
& \ghost{\ShiftMinus}
& \qw
&
&
& \qw
& \qw
& \gate{X} %\targ
& \qw
}
    \end{array}
    \ , \qquad \
    \begin{array}{c}
    \Qcircuit @C=0.8em @R=0.5em {
% & \lstick{\ket{q_2}}
& \multigate{2}{\ShiftPlus}
& \qw
&
&
& \ctrl{2}
& \ctrl{1}
& \gate{X}
& \qw
\\
% & \lstick{\ket{q_1}}
& \ghost{\ShiftPlus}
& \qw
& \push{=}
&
& \ctrl{1}
& \gate{X} % \targ
& \qw
& \qw
\\
% & \lstick{\ket{q_0}}
& \ghost{\ShiftPlus}
& \qw
&
&
& \gate{X} % \targ
& \qw
& \qw
& \qw
}

    \end{array}
    \ .
    \label{eq:circuit-shift}
\end{equation}
In the subsequence, we will need a certain sequence of controlled versions of $\ShiftPlusMinus$, which is treated in the next lemma.
\begin{lemma}
    Let $j \in \{0, 1, \dots, \numDofQC-1\}$ and $\ell \in \{0, 1, 2, 3\}$. For $\ket{\ell} = \ket{\ell_0 \ell_1}$ the circuit
    \begin{equation}
        \begin{array}{c}
        \Qcircuit @C=1em @R=.75em {
        \lstick{\ket{j}}
        & {/} \qw 
        % & \gate{\Shift{-1}}
        & \gate{\ShiftMinus}
        % & \gate{\Shift{+1}}
        & \gate{\ShiftPlus}
        & \qw
        \\
        \lstick{\ket{\ell_1}}
        & \qw
        & \ctrlo{-1}
        & \qw 
        & \qw 
        \\
        \lstick{\ket{\ell_0}}
        & \qw
        & \qw
        & \ctrl{-2}
        & \qw
        }
        \end{array}
        \quad
        \text{ implements }
        \quad
        \ket{\ell} \ket{j}
        \to
        \left\{ \begin{array}{ll}
            \ket{0} \ket{j - 1} \ , & \ell = 0 \ ,
            \\
            \ket{3} \ket{j + 1} \ , & \ell = 3 \ ,
            \\
            \ket{\ell} \ket{j} \ , & \ell = 1 \text{ or } 2 \ .
        \end{array} \right.
    \end{equation}
    \label{lemma:l-r-identity-1d-0-3}
\end{lemma}
\begin{proof}
    Writing $\ell$ in binary representation $\mathtt{l}_0 \mathtt{l}_1$, we see that the circuit maps
    \begin{equation*}
    % \ket{\ell_0 \ell_1} \ket{j}
    \ket{\ell} \ket{j}
        \to
        \left\{ \begin{array}{ll}
            \ket{\bit{0}\bit{0}} \ShiftMinus \ket{j} \ , 
            & \mathtt{l}_0 \mathtt{l}_1 = \mathtt{00} \ ,
            \\
            \ket{\bit{0} \bit{1}} \ket{j} \ , 
            &\mathtt{l}_0 \mathtt{l}_1 = \mathtt{01} \ ,
            \\
            \ket{\bit{1} \bit{0}} \ShiftPlus \ShiftMinus \ket{j} \ , 
            &\mathtt{l}_0 \mathtt{l}_1 = \mathtt{10} \ .
            \\
            \ket{\bit{1}\bit{1}} \ShiftPlus \ket{j} \ , 
            & \mathtt{l}_0 \mathtt{l}_1 = \mathtt{11} \ ,
        \end{array} \right.
    \end{equation*}
    Using \eqref{eq:shift} and the fact that $\ShiftPlus \ShiftMinus = I$ concludes the proof.
\end{proof}
\noindent
Indeed, the circuit in Lemma~\ref{lemma:l-r-identity-1d-0-3} is an efficient implementation of the following multi-controlled circuit
\begin{equation*}
    \begin{array}{c}
    \Qcircuit @C=1em @R=.7em {
        & {/} \qw  & \gate{\ShiftMinus} & \gate{\ShiftPlus} & \qw
        \\
        & \qw & \multimeasure{1}{0} \qwx[-1] & \multimeasure{1}{3} \qwx[-1] & \qw 
        \\
        & \qw & \ghost{0}  & \ghost{3} & \qw
        }
    \end{array}
    \qquad 
    =
    \qquad
    \begin{array}{c}
    \Qcircuit @C=1em @R=.7em {
        & {/} \qw 
        & \gate{\ShiftMinus}
        & \gate{\ShiftPlus}
        & \qw
        \\
        & \qw
        & \ctrlo{-1} 
        & \qw
        & \qw 
        \\
        & \qw
        & \qw
        & \ctrl{-2} 
        & \qw
    }
    \end{array}
    \ .
\end{equation*}
\subsection{Block Encoding}
\label{sec:block-encoding}
As above, let $\numDofQC = 2^\numSytemBits$ be a power of two. For an arbitrary matrix $A \in \C^{\numDofQC \times \numDofQC}$ denote by $\widetilde{A} = \tfrac{1}{\norm{A}_2} A$ its scaled version with unitary norm. The technique of embedding $\widetilde{A}$ into a (larger) unitary matrix $U \in \C^{(\numAncillaDof \numDofQC) \times (\numAncillaDof \numDofQC)}$, $M = 2^\numAncillaBits$, such that
\begin{subequations}
\label{eq:definition-block-encoding}
\begin{equation}
    U
    = \pmat{
        \scalingFactorBE \widetilde{A} & \ast 
        \\ \ast & \ast
    }
    \label{eq:definition-block-encoding-1}
\end{equation}
is called block encoding. Here, $\ast$ denote matrix blocks of appropriate size, whose values are irrelevant, and $\scalingFactorBE \in [-1, 1]$ is a sub-normalization scaling factor. The structure of $U$ consists of $\numAncillaDof$ rows and $\numAncillaDof$ columns of blocks of size $\numDofQC \times \numDofQC$, i.e.
\begin{equation*}
    U \quad = \quad
    \underbrace{
    \arraycolsep=1.4pt
    \begin{array}{ccccc}
        \square & \square & \dots & \square & \square
        \\
        \square & \square & \dots & \square & \square
        \\
        \vdots & \vdots & & \vdots & \vdots
        \\
        \square & \square & \dots & \square & \square
    \end{array}
    }_{\displaystyle \numAncillaDof}
    \left\}
    \begin{array}{l}
    \phantom{\square}
    \\
    \phantom{\square}
    \\
    % {\displaystyle \numAncillaDof}
    \phantom{\vdots}
    \\
    \phantom{\square}
    \end{array}
    \right.
    \hspace{-0.7cm}
    {\displaystyle \numAncillaDof}
    \qquad \quad
    \text{where } \quad
    \square = \numDofQC \times \numDofQC \text{ matrix block}
    \ .
\end{equation*}
Note that the block encoding requires $\numAncillaBits$ ancilla qubits.

Consider the composite state $\ket{w} = \ket{0} \ket{v}$, where $\ket{0} \in \C^\numAncillaDof$ belongs to the ancilla system and $\ket{v} \in \C^\numDofQC$ is an arbitrary vector with unit norm, $\norm{\! \ket{v} \!}_2 = 1$. Then,
\begin{equation}
    U \ket{w}
    = \ket{0} \scalingFactorBE \widetilde{A} \ket{v}
    + \ket{\perp} \ ,
    \label{eq:definition-block-encoding-2}
\end{equation}
where
\begin{equation*}
    \ket{\perp} = \sum_{\ell=1}^{\numAncillaDof - 1}
    \ket{\ell} \ket{\ast} % \ ,
\end{equation*}
is orthogonal to $\ket{0} \scalingFactorBE \widetilde{A} \ket{v}$. This means that if we measure the ancilla qubits in the state $\ket{0}$, the system register contains $\scalingFactorBE \widetilde{A} \ket{v} / \norm{\scalingFactorBE \widetilde{A} \ket{v}}_2$:
\begin{equation}
    \begin{array}{c}
    \Qcircuit @C=1em @R=.3em {
        \lstick{\ket{v}} 
        & {/} \qw 
        & \multigate{2}{U}
        & \qw
        & \qw 
        &\push{\scalingFactorBE \widetilde{A} \ket{v} / \norm{\scalingFactorBE \widetilde{A} \ket{v}}_2}
        \\
        &
        &
        & \push{\ket{0}}
        &
        \\
        \lstick{\ket{0}}
        & {/} \qw
        & \ghost{U}
        % & \meterB{\ket{0}}
        & \meter 
        &
    }
    \end{array}
    \label{eq:definition-block-encoding-3}
\end{equation}
The probability to measure the ancilla system in $\ket{0}$ is called the success probability of the block encoding. It is given by 
\begin{equation}
    \mathrm{success\ probability}
    = \scalingFactorBE^2 \norm{\widetilde{A} \ket{v}}_2^2 \ .
    \label{eq:definition-block-encoding-4}
\end{equation}
\end{subequations}
%
% $\scalingFactorBE^2 \norm{\widetilde{A} \ket{v}}_2^2$.
Note that the success probability depends on the sub-normalization factor $\scalingFactorBE$ and is the highest for $|\scalingFactorBE|=1$. Also note that the sub-normalization factor plays an important role if $U$ is used in methods for the Hamiltonian simulation problem, where an initial state $\ket{v_0}$ evolves over time as $\mathrm{e}^{\i t \widetilde{A}} \ket{v_0}$. Here, the sub-normalization factor multiplies the simulation time $t$ due to $\mathrm{e}^{\i t \widetilde{A}} = \mathrm{e}^{\i (t/\scalingFactorBE) (\scalingFactorBE \widetilde{A})}$.
\section{Finite Difference Discretization of the Laplacian}
\label{sec:finite-differences}
Let $\domain_\numDim = [0, 1]^\numDim$ be the $\numDim$-dimensional unit hypercube. We denote the components of points $\vecx \in \domain_\numDim$ by $\vecx = (x^{(0)}, x^{(1)}, \dots, x^{(\numDim-1)})$. On $\domain_\numDim$ we consider the Laplacian $\laplace_\numDim$ defined by
\begin{subequations}
\label{eq:laplace-d-dim}
\begin{equation}
    \laplace_\numDim
    = \sum_{d=0}^{\numDim-1} 
    \tfrac{\partial^2}{\partial (x^{(d)})^2} \ ,
    \label{eq:laplace-d-dim-diff-op}
\end{equation}
for functions $v$ with periodic boundary conditions
\begin{equation}
    v(x^{(0)}, \dots, \underbrace{0}_d, \dots, x^{(\numDim-1)})
    = 
    v(x^{(0)}, \dots, \underbrace{1}_d, \dots, x^{(\numDim-1)})
    \label{eq:laplace-d-dim-bnd-cond}
\end{equation}
for all $d = 0, 1, \dots, \numDim-1$. The periodicity condition \eqref{eq:laplace-d-dim-bnd-cond} means that $v$ has the same value on all opposite sides of $\domain_\numDim$.
\end{subequations}

Let
\begin{equation*}
    \domain_{1,h}
    = \{ j h \ | \ j = 0, 1, \dots, \numDofOne-1\}
\end{equation*}
be a discretization of the unit interval $\domain_1 = [0, 1]$ into $\numDofOne$ equidistant points with grid width $h = 1/\numDofOne$. We assume that $\numDofOne$ is a power of two, i.e. $\numDofOne = 2^\numSytemBits$ for some $\numSytemBits \in \N$. Note that in $\domain_{1,h}$ we exclude the end point $x=1$ due to the periodic boundary conditions. Using $\domain_{1,h}$ we can build a discretization of $\domain_{\numDim}$ with $\numDof = \numDofOne^\numDim = 2^{\numSytemBits \numDim}$ equidistant points by
\begin{align*}
    \domain_{\numDim,h}
    &= \domain_{1,h} \times \domain_{1,h} \times \dots \domain_{1,h}
    \\
    &= \bigl\{ (\jind{0} h , \jind{1} h, \dots, \jind{\numDim-1} h)
    \ | \ 
    % \jind{0}, \jind{1}, \dots, \jind{\numDim-1} = 0, 1, \dots, \numDofOne-1 
    \jind{d} = 0, 1, \dots, \numDofOne-1 
    \bigr\} \ .
\end{align*}
In Figure~\ref{fig:grid-examples} we give examples for $\domain_{\numDim,h}$ with dimensions $\numDim = 1, 2,$ and $3$.
\begin{figure}
    \centering
    \begin{subfigure}{0.3\linewidth}
        \includegraphics[width=\gridPlotWidth]{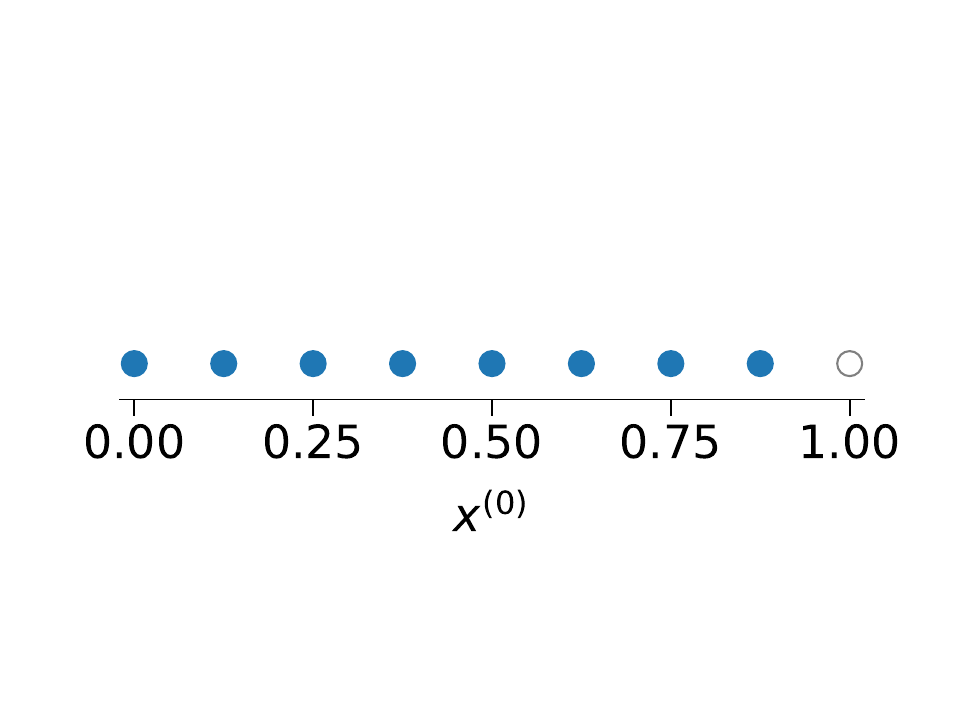}
        \\[0.5cm]
        \caption{Grid $\domain_{1,h}$ with $\numDofOne = 8$, i.e. $\numDofQC_1 =\numDofOne = 8$ points.}
    \end{subfigure}
    \ \ 
    \begin{subfigure}{0.3\linewidth}
        \includegraphics[width=\gridPlotWidth]{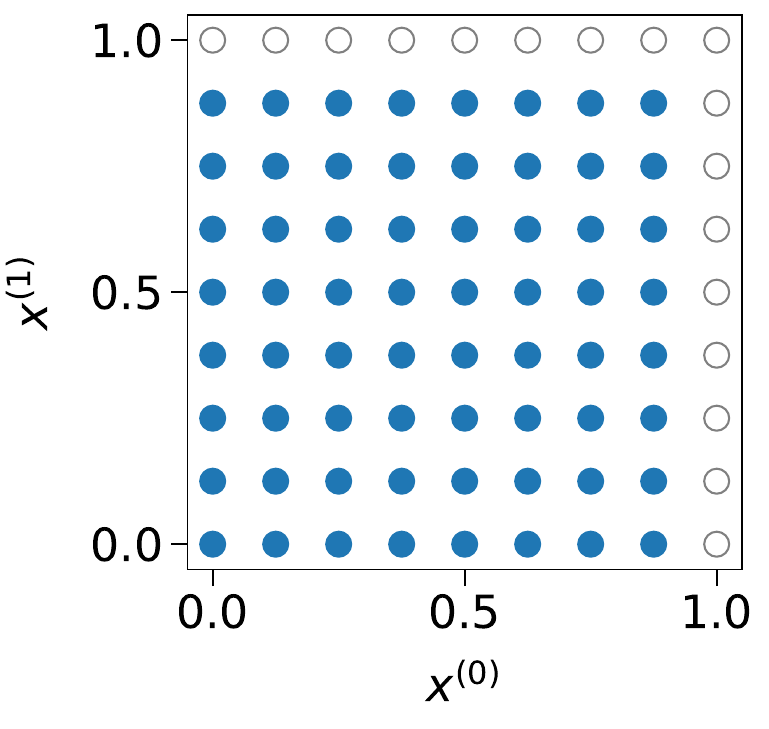}
        \caption{Grid $\domain_{2,h}$ with $\numDofOne = 8$, i.e. $\numDofQC_2 =\numDofOne^2 = 64$ points.}
    \end{subfigure}
    \ \ 
    \begin{subfigure}{0.3\linewidth}
        \includegraphics[width=\gridPlotWidth]{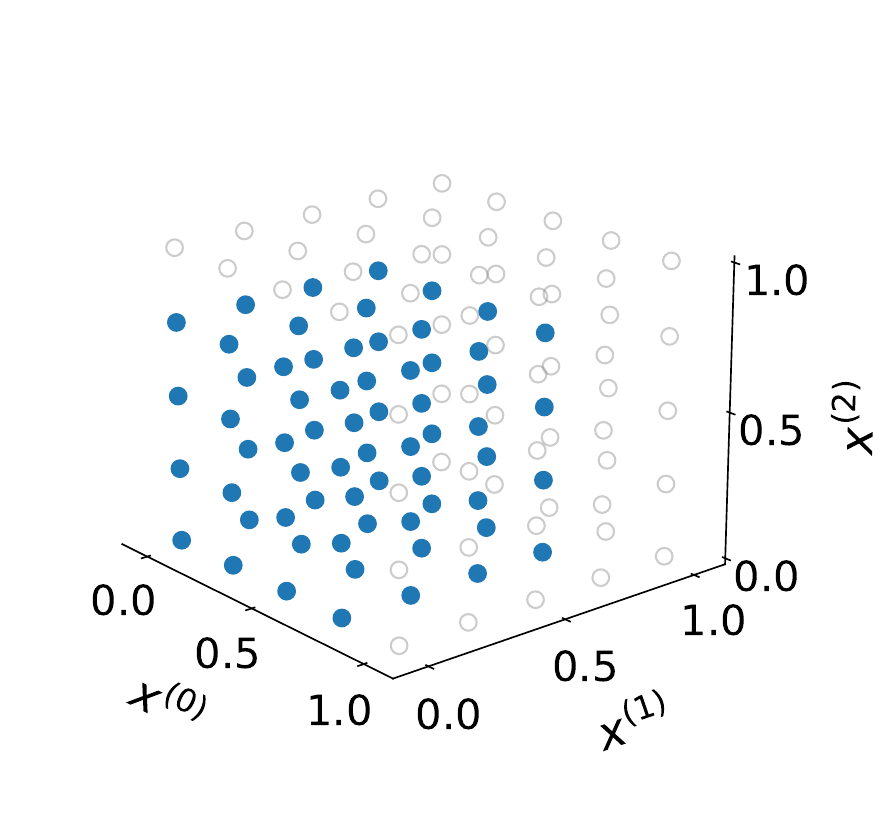}
        \caption{Grid $\domain_{3,h}$ with $\numDofOne = 4$, i.e. $\numDofQC_3 =\numDofOne^3 = 64$  points.}
    \end{subfigure}
    \caption{Examples of grids $\domain_{\numDim,h}$ for the domain $\domain_{\numDim} = [0, 1]^\numDim$ for dimensions $\numDim=1$ (a), $\numDim=2$ (b), and $\numDim=3$ (c). The solid blue points are the grid points $\vecx_k \in \domain_{\numDim,h}$. The grey circled points do not belong to the grid as the value of a function $v$ on those points is determined by the periodic boundary condition \eqref{eq:laplace-d-dim-bnd-cond}.}
    \label{fig:grid-examples}
\end{figure}
We map the set $\domain_{\numDim,h}$ into a vector $\vecDomain_{\numDim,h}$ by
\begin{equation*}
    \vecDomain_{\numDim,h}
    = (\jind{0} h , \jind{1} h, \dots, \jind{\numDim-1} h) 
    \ketSimple{\jind{\numDim-1}} \dots \ketSimple{\jind{1}} \ketSimple{\jind{0}}
    \ .
\end{equation*}
The corresponding ordering of the grid points for $\vecDomain_{2,h}$ can be seen in
Figure~\ref{fig:grid-examples-vectorization}.
\begin{figure}
    \centering
    \includegraphics[width=\gridPlotWidth]{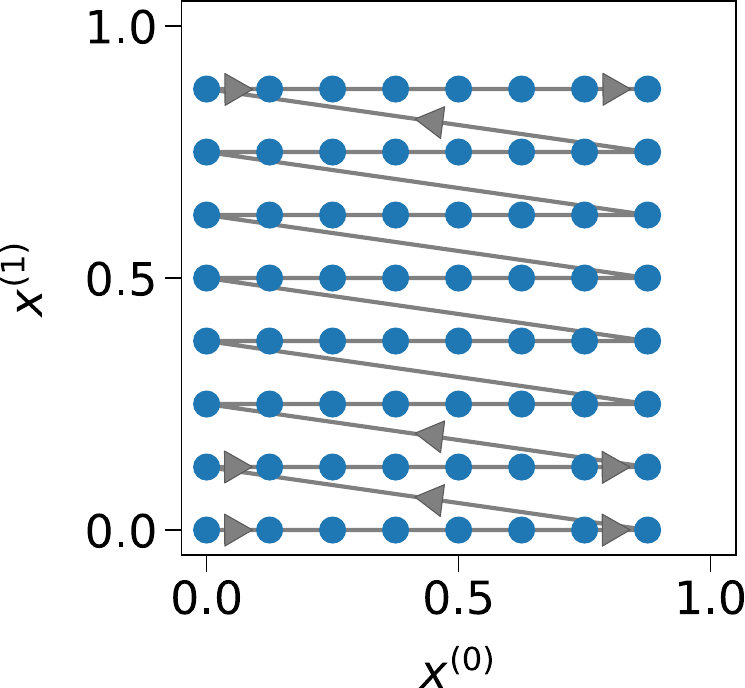}
    \caption{Example of the ordering of the grid points of the vectorization of $\domain_{2,h}$ into $\vecDomain_{2,h}$.}
    \label{fig:grid-examples-vectorization}
\end{figure}
In this setting, the finite difference approximation of the one dimensional Laplacian $\laplaceOne$ with periodic boundary conditions is given by
\begin{equation}
    \laplaceOneh
    = \frac{1}{h^2}
    \arraycolsep=1.4pt
    \pmat{
    -2 & 1  &        &         &       & 1            
    \\
    1  & -2 & 1
    \\
    %    &  1 & -2     & 1
    % \\
       &    & \ddots & \ddots & \ddots
    \\
       &    &        &    1   &  -2    & 1
    \\
    1  &    &        &        &  1     & -2
    } \ .
\end{equation}
For higher dimensions we have, due to the tensor structure of $\domain_{\numDim, h}$ and due to the ordering of $\vecDomain_{\numDim,h}$,
\begin{equation}
    % \laplace_{\numDim, h}
    \laplaceDh
    = \sum_{d=0}^{\numDim-1}
    \underbrace{
        I_{\numDofOne} \otimes \dots \otimes I_{\numDofOne}
    }_{
        \numDim-1-d \text{ times}
    } 
    \ \otimes \ \laplaceOneh \ \otimes \
    \underbrace{
        I_{\numDofOne} \otimes \dots \otimes I_{\numDofOne}
    }_{
        d \text{ times}
    } \ .
    \label{eq:laplaceDh}
\end{equation}
Here, $I_{\numDofOne}$ is the $\numDofOne \times \numDofOne$ dimensional identity matrix. 
The eigenvalue with largest absolute value of $\laplace_{\numDim, h}$ is given by 
\begin{subequations}
    \begin{equation}
    % \lambda_{\numDim, \mathrm{max}} = \tfrac{4 \numDim}{h^2} \ ,
    \lambdaMax{\numDim} = \frac{4 \numDim}{h^2} \ ,
    \end{equation}
so that the scaled version
\begin{equation}
    \laplaceDhScaled
    = \frac{1}{\lambdaMax{\numDim}}
    \laplaceDh
    \label{eq:laplaceDhScaled-definition}
\end{equation}
\end{subequations}
has unitary norm $\norm{\widetilde{\laplace}_{\numDim, h}}_2 = 1$. Inserting this into \eqref{eq:laplaceDh} we obtain
\begin{equation}
    \laplaceDhScaled
    = \frac{1}{\numDim}
    \sum_{d=0}^{\numDim-1}
    \underbrace{
        I_{\numDofOne} \otimes \dots \otimes I_{\numDofOne}
    }_{
        \numDim-1-d \text{ times}
    } 
    \ \otimes \ \laplaceOnehScaled \ \otimes \
    \underbrace{
        I_{\numDofOne} \otimes \dots \otimes I_{\numDofOne}
    }_{
        d \text{ times}
    } \ .
    \label{eq:laplaceDhScaled}
\end{equation}
Let us remark that for $\numDim = 1$ and the canonical basis vector $\ket{j}$, $j=0, 1, \dots, \numDofOne-1$, we have
\begin{equation}
    \laplaceOnehScaled \ket{j}
    =
    \tfrac{1}{4} \ket{j - 1} \
    - \tfrac{1}{2} \ket{j} \
    + \tfrac{1}{4} \ket{j + 1}
    \ .
    \label{eq:laplace-1d-applied-to-canonical-basis}
\end{equation}
Recall at this point our convention to drop $\mod \numDofOne$ in $\ket{j \pm 1}$. From \eqref{eq:laplace-1d-applied-to-canonical-basis} the action of $\laplaceDhScaled$ for general $\numDim$ on basis vectors $\ket{j}$ can easily be calculated. For example, for $\numDim = 2$ and $\ket{j} = \ketSimple{j^{(1)}}\ketSimple{j^{(0)}}$, $j^{(d)} = 0, 1, \dots, \numDofOne-1$, we have
\begin{align*}
    \laplaceTwohScaled \ket{j}
    =& \tfrac{1}{2} \Bigl(
    \ketSimple{j^{(1)}} \laplaceOnehScaled  \ketSimple{j^{(0)}}
    +
    \laplaceOnehScaled  \ketSimple{j^{(1)}}\ketSimple{j^{(0)}}
    \Bigr)
    \\
    =& \tfrac{1}{8} \Bigl(
    \ketSimple{j^{(1)}}\ketSimple{j^{(0)}-1}
    + \ketSimple{j^{(1)}+1}\ketSimple{j^{(0)}}
    \Bigr)
    \\
    &- \tfrac{1}{2} \ketSimple{j^{(1)}}\ketSimple{j^{(0)}}
    % \\
    % &
    + \tfrac{1}{8} \Bigl(
    \ketSimple{j^{(1)}-1}\ketSimple{j^{(0)}}
    + \ketSimple{j^{(1)}+1}\ketSimple{j^{(0)}}
    \Bigr) \ .
    \label{eq:laplace-2d-applied-to-canonical-basis}
\end{align*}
We conclude this section with a lemma that will later allow us to calculate success probabilities. 
\begin{lemma}
    Let $v: \domain_\numDim \to \domain_\numDim$ be a four times differentiable function with periodic boundary conditions. Denote by $\vecv = (v(\vecx_k))_{k=0}^{\numDof-1}$ the vector of evaluations of $v$ on the grid points $\vecx_k \in \vecDomain_{\numDim,h}$. Then, for $\ket{v} = \vecv / \norm{\vecv}_2$ we have
    \begin{equation}
        \laplaceDhScaled \ket{v}
        % = \tfrac{h^2}{4 \numDim \norm{\vecv}_2} 
        = \frac{1}{\lambda_{\numDim, \mathrm{max}} \norm{\vecv}_2} 
        \Bigl(
            \laplaceD \vecv
            + \mathcal{O}(h^2) 
        \Bigr) \ ,
        \qquad
        \laplaceD \vecv 
        = \bigl(
            (\laplaceD v)(\vecx_k)
        \bigr)_{k=0}^{\numDof-1} \ .
        \label{eq:success-prob-laplace-1}
    \end{equation}
    For small $h$ we have the scaling
    \begin{equation}
        \bigl\| \laplaceDhScaled \ket{v} \bigr\|_2
        \sim
        C(\numDim, v) h^2 \ ,
        \qquad
        C(\numDim, v)
        = \frac{1}{4 \numDim}
        \frac{
            \norm{\laplaceD v}_{L^2(\domain_\numDim)}
        }{
            \norm{v}_{L^2(\domain_\numDim)}
        } \ .
        \label{eq:success-prob-laplace-2}
    \end{equation}
    \label{lemma:success-prob-laplace}
\end{lemma}
\begin{proof}
    Taylor expansion for a function $u$ around $x_{k \pm 1} = x_k \pm h$ gives
    \begin{equation}
        u(x_{k + 1}) - 2u(x_{k}) + u(x_{k - 1})
        = h^2 u^{\prime \prime}(x_k) 
        + R_3(h) \ , 
        \quad
        |R_3(h)| \le \tfrac{h^4}{12}
        \max_{x \in [0, 1]} |u^{(4)}(x)|
        \ ,
        % + \mathcal{O}(h^4) \ ,
        \label{eq:taylor-expansion}
    \end{equation}
    and consequently
    \begin{equation*}
        \laplaceDh \vecv
        = \laplaceD \vecv + \mathcal{O}(h^2) \ .
    \end{equation*}
    The assertion \eqref{eq:success-prob-laplace-1} then follows with \eqref{eq:laplaceDhScaled-definition}.
    For the statement \eqref{eq:success-prob-laplace-2} note that for a function $w: \Omega_\numDim \to \Omega_\numDim$ the 2-norm of its evaluations $\vecw = (w(\vecx_k))_k$ on the grid $\vecDomain_{\numDim,h}$ can be written as
    \begin{equation*}
        \norm{\vecw}_2^2
        = \frac{1}{h^{\numDim}} 
        %\sum_{k=0}^{\numDof-1} h^2 |v(x_k)|^2
        \sum_{\vecx_k \in \vecDomain_{\numDim,h}}
        h^{\numDim} |w(\vecx_k)|^2
        \approx \frac{1}{h^{\numDim}}
        \int_{\Omega_\numDim} |w(\vecx)|^2 \ d\vecx
        = \frac{1}{h^{\numDim}}
        \norm{w}_{L^2(\Omega_\numDim)}^2 \ .
    \end{equation*}
    Here, the approximation is due to the Riemann summation. Inserting this into \eqref{eq:success-prob-laplace-1} concludes the proof.
\end{proof}
\section{Block Encoding of One Dimensional Laplace}
\label{sec:block-encoding-laplace-1d}
In this section we show the block encoding of $\laplaceOnehScaled$. This will allow us to derive block encodings of general $\laplaceDhScaled$ in the next section.
\begin{theorem}
The circuit
\begin{subequations}
\label{eq:theorem-encoding-laplace-1d}
\begin{equation}
    \begin{array}{c}
    \Qcircuit @C=1em @R=.7em {
    \lstick{\ket{j}} 
    & {/} \qw 
    & \barrier[-1.5em]{2} \qw 
    & \gate{\ShiftMinus} 
    & \gate{\ShiftPlus} \barrier[-1.2em]{2} 
    & \qw 
    & \qw 
    % &
    \\
    \lstick{\ket{\ell_1}} 
    & \gate{H} 
    & \gate{Z} 
    & \ctrlo{-1}
    & \qw 
    & \gate{H} 
    % & \meter
    & \qw
    \\
    \lstick{\ket{\ell_0}} 
    & \gate{H} 
    & \gate{Z} 
    & \qw 
    & \ctrl{-2} 
    & \gate{H} 
    % & \meter
    & \qw
    \\
    \\
    \\
    & \text{part 1} 
    &
    & \text{part 2}
    &
    & \text{part 3}
}
    \end{array}
    \label{eq:theorem-encoding-laplace-1d-1}
\end{equation}
is a block encoding of $\laplaceOnehScaled$ in the sense of 
\eqref{eq:definition-block-encoding} with $\numAncillaBits=2$ and $\scalingFactorBE=1$, i.e. 
\begin{equation}
    U 
    = \pmat{\laplaceOnehScaled & \ast \\ \ast & \ast} 
    \in \R^{4 \numDofOne \times 4 \numDofOne} \ .
    \label{eq:theorem-encoding-laplace-1d-2}
\end{equation}
If $U$ is applied to a state $\ket{v}$ satisfying the assumptions of Lemma~\ref{lemma:success-prob-laplace} the success probability scales as
\begin{equation}
    \bigl\| \laplaceOnehScaled \ket{v} \bigr\|_2^2
    \sim
    \frac{h^4}{16}
    \frac{
            \norm{\laplaceD v}_{L^2(\domain_1)}^2
        }{
            \norm{v}_{L^2(\domain_1)}^2
    } \ .
    \label{eq:theorem-encoding-laplace-1d-3}
\end{equation}
\end{subequations}
\label{thm:block-encoding-laplace-1d}
\end{theorem}
\begin{proof}
We show that \eqref{eq:definition-block-encoding-2} is satisfied for any basis vector $\ket{j}$, $j = 0, 1, \dots, \numDofOne-1$. Going through the three parts of the quantum circuit \eqref{eq:theorem-encoding-laplace-1d-1} yields:
\begin{align*}
    \ket{0} \ket{j}
    &\to \tfrac{1}{\sqrt{2^2}} \Bigl(
    \ket{0} - \ket{1} - \ket{2} + \ket{3}
    \Bigr) \ket{j}
    \\
    % &\stackrel{Lemma~\ref{lemma:l-r-identity-1d-0-3}}{\to}
    &\to
    \tfrac{1}{\sqrt{2^2}} \Bigl(
    \ket{0}\ket{j - 1}
    - \ket{1} \ket{j}
    - \ket{2} \ket{j}
    + \ket{3}\ket{j + 1}
    \Bigr)
    \\
    &\to \tfrac{1}{2^2} 
    \ket{0} \Bigl( \ket{j-1} - 2 \ket{j} + \ket{j+1} \Bigr)
    \\
    & \quad + \tfrac{1}{2^2} 
    \ket{1} \Bigl( \ket{j-1} - \ket{j+1} \Bigr)
    \\
    & \quad + \tfrac{1}{2^2} 
    \ket{2} \Bigl( \ket{j-1} - \ket{j+1} \Bigr)
    \\
    & \quad + \tfrac{1}{2^2} 
    \ket{3} \Bigl( \ket{j-1} + 2 \ket{j} + \ket{j+1} \Bigr)
    \numberthis \label{eq:proof-block-encoding-laplace-1d}
    \\
    &= \ket{0} \laplaceOnehScaled \ket{j} 
    + \sum_{\ell=1}^3 \ket{\ell} ( \dots ) \ .
\end{align*}
Here, we used Lemma~\ref{lemma:l-r-identity-1d-0-3} for part 2.
\end{proof}
\noindent
From \eqref{eq:proof-block-encoding-laplace-1d} we obtain the first column of blocks of $U$ in \eqref{eq:theorem-encoding-laplace-1d-2}. Doing the same calculations for $\ket{\ell}\ket{j}$ with $\ell=1, 2,$ and $3$ as in the proof of Theorem~\ref{thm:block-encoding-laplace-1d}, one can obtain the other three columns and show that
\begin{equation*}
    U = \pmat{
    \laplaceOnehScaled & \DOnehScaledTwo & \DOnehScaledTwo & \SOnehScaled
    \\
    \DOnehScaledTwo & \laplaceOnehScaled & \SOnehScaled & \DOnehScaledTwo
    \\
    \DOnehScaledTwo & \SOnehScaled & \laplaceOnehScaled & \DOnehScaledTwo 
    \\
    \SOnehScaled & \DOnehScaledTwo & \DOnehScaledTwo & \laplaceOnehScaled
    } \ .
\end{equation*} 
Here, $\DOnehScaledTwo = \tfrac{h}{2} \DOneh$ is a scaled finite difference approximation to $\tfrac{d}{dx}$ and $\SOnehScaled = \tfrac{1}{2h} \SOneh$ is a scaled version of the trapezoidal quadrature rule that approximates $\int_0^1 \cdot \; dx$. We have
\begin{equation}
    \DOneh
    = \tfrac{1}{2 h}
    \arraycolsep=1.4pt
    \pmat{
    0 & 1  &        &         &       & -1            
    \\
    -1  & 0 & 1
    \\
       &    & \ddots & \ddots & \ddots
    \\
       &    &        &    1   &  0    & 1
    \\
    -1  &    &        &        &  1     & 0
    } \ ,
    \qquad
    \SOneh
    = \tfrac{h}{2}
    \pmat{
    1 & 2  &        &         &       & 1
    \\
    1  & 2 & 1
    \\
       &    & \ddots & \ddots & \ddots
    \\
       &    &        &    1   &  2    & 1
    \\
    1  &    &        &        &  1     & 2
    } \ .
    \label{eq:doneh}
\end{equation}
In Figure~\ref{fig:heat-map-laplace-1d} we give a heatmap plot of the entries of the unitary $U$.
\begin{figure}
    \centering
    \begin{subfigure}{0.4\textwidth}
        \includegraphics[width=\heatmapPlotWidth]{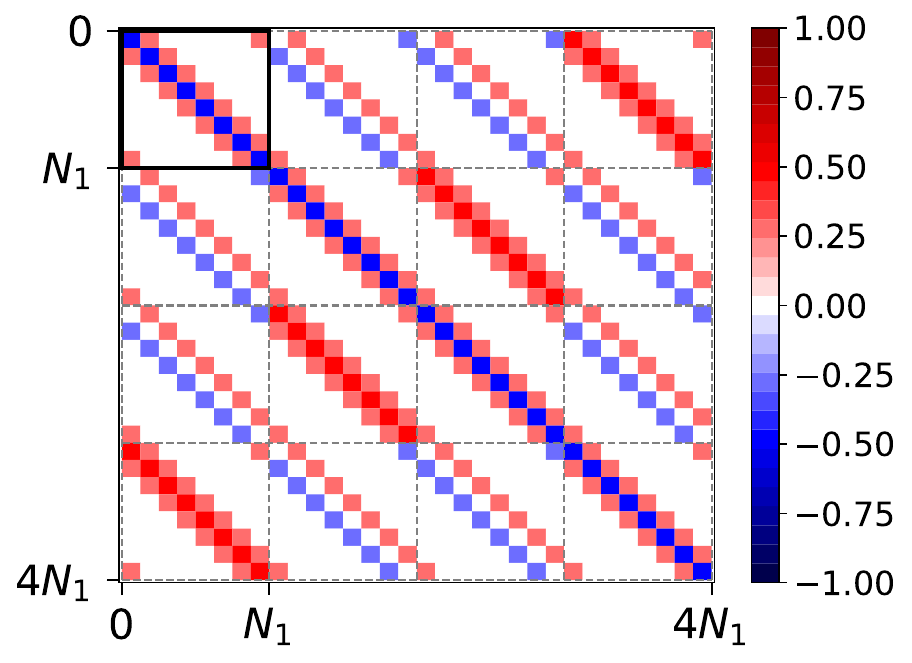}
        \caption{Entire matrix $U$.}
    \end{subfigure}
    \qquad
    \begin{subfigure}{0.5\textwidth}
        \includegraphics[width=\heatmapPlotWidth]{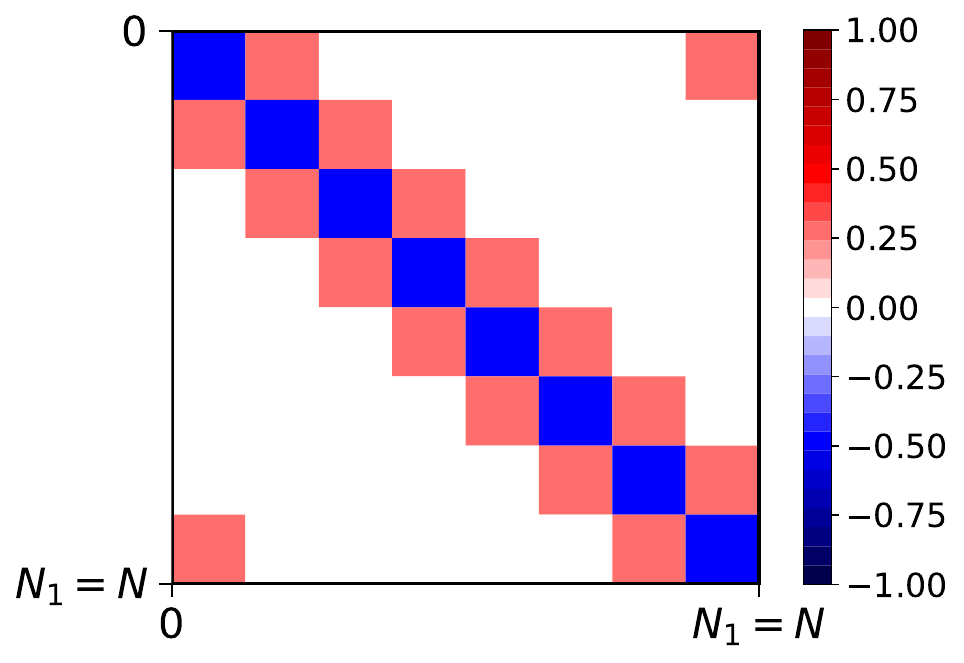}
        \caption{$(1,1)$-block of $U$. This block corresponds to $\laplaceOnehScaled$.}
    \end{subfigure}
    \caption{
     Heat map plot of the entries of the matrix $U$ in \eqref{eq:theorem-encoding-laplace-1d-1} in Theorem~\ref{thm:block-encoding-laplace-1d} for $\numSytemBits=3$, which corresponds to 
     % $\numDofOne = 2^3 = 8$ 
     $\numDofQC_1 = \numDofQC = 2^3 = 8$ 
     grid points. The entire matrix is displayed in (a), where $\laplaceOnehScaled$ is encoded in the $(1,1)$-block marked with bold lines. A zoom into the $(1,1)$-block is given in (b).}
    \label{fig:heat-map-laplace-1d}
\end{figure}
\subsection{Comparison with other block encodings}
\label{sec:other-block-encodings}
In \cite{Camps.2022} it is shown that real matrices $A$ of the form 
\begin{equation*}
    A
    = 
    \pmat{
    a_0 & a_{-1}  &        &                & a_{1}            
    \\
    a_{1}  & a_0 & a_{-1}
    \\
        & \ddots & \ddots & \ddots
    \\
        &        &    a_{1}   &  a_0   & a_{-1}
    \\
    a_{-1}     &        &        &  a_{1}     & a_0
    } \ ,
\end{equation*}
with $\norm{A}_2 \le 1$ and $a_0 > 0$, can be encoded by the circuit
\begin{subequations}
\label{eq:camps-circuit-laplace-1d}
\begin{equation}
    \begin{array}{c}
    \Qcircuit @C=1em @R=.5em {
\lstick{\ket{j}} 
& {/} \qw 
& \qw
& \qw
& \qw
& \gate{\ShiftMinus} 
& \gate{\ShiftPlus}
& \qw 
& \qw 
% &
\\
\lstick{\ket{\ell_1}} 
& \gate{H} 
& \ctrlo{1}
& \ctrlo{1}
& \ctrl{1}
& \ctrl{-1}
& \qw 
& \gate{H} 
% & \meter
& \qw
\\
\lstick{\ket{\ell_0}} 
& \gate{H} 
& \ctrlo{1}
& \ctrl{1}
& \ctrlo{1}
& \qw 
& \ctrl{-2} 
& \gate{H} 
% & \meter
& \qw
\\
\lstick{\ket{a}}
& \qw
& \gate{RY(\theta_0)}
& \gate{RY(\theta_1)}
& \gate{RY(\theta_2)}
& \qw
& \qw
& \qw
& \qw
}
    \end{array}
    \label{eq:camps-circuit-laplace-1d-1}
\end{equation}
Here, the angles $\theta_k$ are defined by
\begin{equation}
    \theta_0 = 2 \arccos(a_0-1) \ ,
    \quad
    \theta_1 = 2 \arccos a_{1} \ ,
    \quad
    \theta_2 = 2 \arccos a_{-1} \ .
     \label{eq:camps-circuit-laplace-1d-2}
\end{equation}
In fact, the circuit \eqref{eq:camps-circuit-laplace-1d} yields the block encoding
\begin{equation}
    U = \pmat{
        \tfrac14 A & \ast
        \\
        \ast & \ast
    } \ .
    \label{eq:camps-circuit-laplace-1d-3}
\end{equation}
\end{subequations}
Choosing $a_0 = \tfrac12$ and $a_{1} = a_{-1} = -\tfrac14$ allows to encode $- \tfrac14 \laplaceOnehScaled$. This corresponds to a block encoding of $\laplaceOnehScaled$ in the sense of \eqref{eq:definition-block-encoding} with $\numAncillaBits=3$ and $\scalingFactorBE=-\tfrac{1}{4}$. In Figure~\ref{fig:heat-map-laplace-1d-camps} we provide a heatmap plot of this block encoding.

Let us point out the following differences of this block encoding in contrast to the one proposed in Theorem~\ref{thm:block-encoding-laplace-1d}. First, due to the positivity condition on $a_{0}$, only the negative version of $\laplaceOnehScaled$ can be encoded. Secondly, by \eqref{eq:definition-block-encoding-4} the scaling factor $\scalingFactorBE=-\tfrac{1}{4}$ decreases the success probability by a factor of 16. And thirdly, the additional ancillary qubit and the multi-controlled $RY$ gates increase the required quantum computing resources. This will further be discussed in Section~\ref{sec:resource-estimates}.
\begin{figure}
    \centering
    \begin{subfigure}{0.4\textwidth}
        \includegraphics[width=\heatmapPlotWidth]{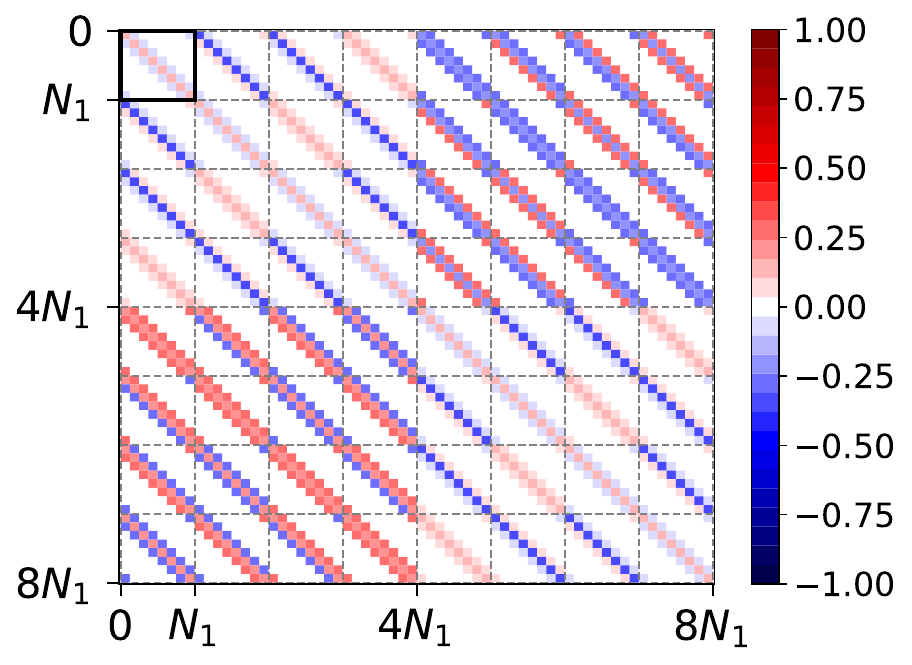}
        \caption{Entire matrix $U$.}
    \end{subfigure}
    \qquad
    \begin{subfigure}{0.5\textwidth}
        \includegraphics[width=\heatmapPlotWidth]{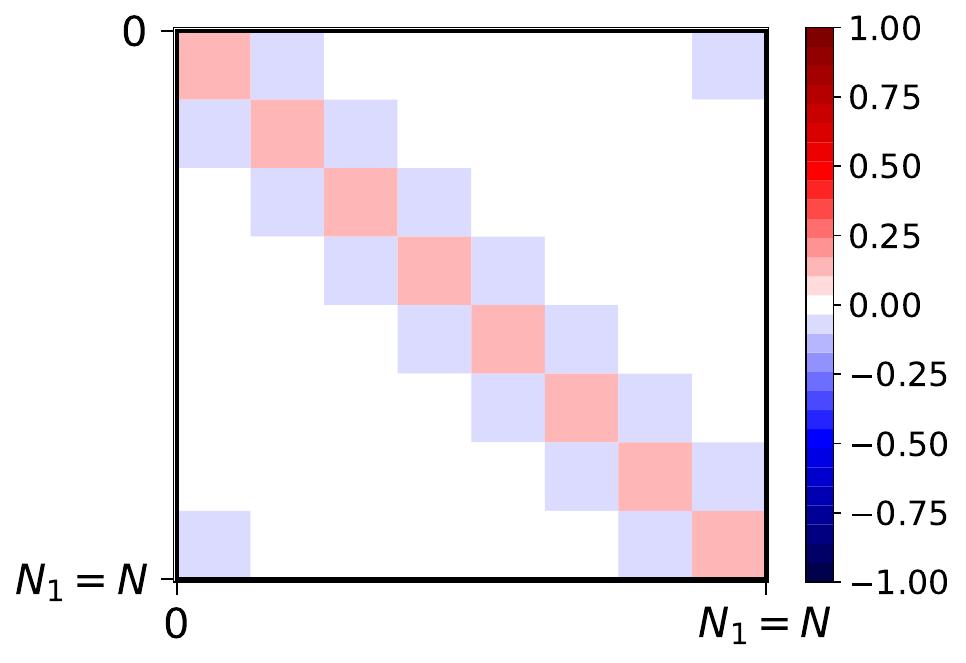}
        \caption{$(1,1)$-block of $U$. This corresponds to $- \tfrac{1}{4} \laplaceOnehScaled$.}
    \end{subfigure}
    \caption{Heat map plot of the entries of the matrix $U$ in \eqref{eq:camps-circuit-laplace-1d-3} for $\numSytemBits=3$, which corresponds to 
    % $\numDofOne = 2^3 = 8$ 
    $\numDofQC_1 = \numDofQC = 2^3 = 8$ 
    grid points. 
    The entire matrix is displayed in (a), where $- \tfrac{1}{4} \laplaceOnehScaled$ is encoded in the $(1,1)$-block marked with bold lines. A zoom into the $(1,1)$-block is given in (b).}
    \label{fig:heat-map-laplace-1d-camps}
\end{figure}
\section{Block Encoding of Arbitrary Dimensional Laplace}
\label{sec:block-encoding-laplace}
In the next theorem we show a block encoding of $\laplaceDhScaled$ for arbitrary dimensions $\numDim$. Interestingly, a sub-normalization scaling is applied to $\laplaceDhScaled$ in the cases when $\numDim$ is a not power of two, $\numDim \neq 2^{\numAncillaKBits}$ for some $\numAncillaKBits \in \N$.
\begin{theorem}
\label{thm:block-encoding-laplace-arbitrary-dim}
Let $\numDim > 1$. Define $\numAncillaKBits = \lceil \log \numDim \rceil$ and $\numAncillaKDof = 2^{\numAncillaKBits}$. Let $\ket{k}$ be an ancilla register with $\numAncillaKBits$ qubits and $H_{\numAncillaKBits} = H^{\otimes \numAncillaKBits}$. Then, the circuit 
\begin{subequations}
\label{eq:circuit-laplace-arbitrary-dim}
\begin{equation}
    \centerline{
    \Qcircuit @C=1em @R=.5em {
    \lstick{\ket{\jind{0}}} 
    & {/} \qw 
    & \qw \barrier[-1.2em]{5}
    & \qw 
    & \qw
    & \gate{\ShiftMinus} 
    & \gate{\ShiftPlus} 
    & \qw 
    & \qw
    & \qw 
    & \qw \barrier[-.2em]{5}
    & \qw
    & \qw
    \\
    \push{\vdots}
    & 
    & 
    & 
    & 
    & 
    & 
    & \push{\ddots}
    &           
    & 
    \\
    \lstick{\ket{\jind{\numDim-1}}} 
    & {/} \qw 
    & \qw 
    & \qw 
    & \qw 
    & \qw 
    & \qw
    & \qw
    & \gate{\ShiftMinus} 
    & \gate{\ShiftPlus} 
    & \qw
    & \qw 
    & \qw
    \\
    \lstick{\ket{\ell_1}} 
    & \qw 
    & \qw 
    & \gate{H} 
    & \gate{Z}
    & \ctrlo{-3} 
    & \qw
    & \push{\ \dots \ } \ghost{\dots} 
    & \ctrlo{-1} 
    & \qw
    & \gate{H}
    & \qw 
    & \qw
    \\
    \lstick{\ket{\ell_0}} 
    & \qw 
    & \qw 
    & \gate{H} 
    & \gate{Z}
    & \qw
    & \ctrl{-4} 
    & \push{\ \dots \ } \ghost{\dots} 
    & \qw
    & \ctrl{-2}
    & \gate{H} 
    & \qw 
    & \qw
    \\
    \lstick{\ket{k}} 
    & {/} \qw
    & \gate{H_{\numAncillaKBits}} 
    & \qw 
    & \qw
    & \measure{0} \qwx[-2] %\ctrl{-2} 
    & \measure{0} \qwx[-2] %\ctrl{-2}
    & \push{\ \dots \ } \ghost{\dots} 
    & \measure{\scriptstyle\numDim-1} \qwx[-2] %\ctrl{-2}
    & \measure{\scriptstyle\numDim-1} \qwx[-2] %\ctrl{-2}
    & \qw
    & \qw
    & \gate{H_{\numAncillaKBits}} 
    & \qw
    \\
    \\
    \\
    & \text{part 1} 
    &
    & 
    & 
    &
    &
    & \text{part 2}
    &
    &
    &
    &
    & \text{part 3}        
    }
}
    \label{eq:circuit-laplace-arbitrary-dim-1}
\end{equation}
is a block encoding of $\laplaceDhScaled$ in the sense of \eqref{eq:definition-block-encoding} with 
\begin{equation}
    \scalingFactor
    = \frac{\numDim}{\numAncillaKDof}
    \qquad
    \text{and}
    \qquad
    \numAncillaBits=2+\numAncillaKBits \ .
    \label{eq:circuit-laplace-arbitrary-dim-3}
\end{equation}
This means, we have
\begin{equation}
    U
    = \pmat{
        \scalingFactor \laplaceDhScaled 
        & \ast
        \\ 
        \ast 
        &
        \ast
    }
    \in \R^{4 \numAncillaKDof \numDof \times 4 \numAncillaKDof \numDof} 
     \ .
    \label{eq:circuit-laplace-arbitrary-dim-2}
\end{equation}
If $U$ is applied to a state $\ket{v}$ satisfying the assumptions of Lemma~\ref{lemma:success-prob-laplace} the success probability scales as
\begin{equation}
    \bigl\| \scalingFactor \laplaceDhScaled \ket{v} \bigr\|_2^2
    \sim
    % \frac{h^4}{2^{2\numAncillaKBits+2}}
     \frac{h^4}{16 \numAncillaKDof^2}
    \frac{
            \norm{\laplaceD v}_{L^2(\domain_\numDim)}^2
        }{
            \norm{v}_{L^2(\domain_\numDim)}^2
    } \ .
    \label{eq:circuit-laplace-power-2d-4}
\end{equation}
\end{subequations}
\end{theorem}
\begin{proof}
We prove that \eqref{eq:definition-block-encoding-2} is satisfied for any basis vector $\ket{j} = \ketSimple{\jind{\numDim-1}} \dots \ketSimple{\jind{1}} \ketSimple{\jind{0}}$, $\jind{d} = 0, 1, \dots, \numDofOne-1$. The three parts of the quantum circuit \eqref{eq:circuit-laplace-arbitrary-dim-1} yield the following transformations:
\begin{align*}
    \ket{0} \ket{0} \ket{j}
    \to&
    \tfrac{1}{\sqrt{\numAncillaKDof}} 
    \sum_{k=0}^{\numAncillaKDof-1}
    \ket{k} \ket{0} \ket{j}
    \\
    \to&
    \tfrac{1}{\sqrt{\numAncillaKDof}}
    % \tfrac{1}{\sqrt{2^{\numAncillaKBits}}}
    \ket{0} \ket{0}
    \ketSimple{\jind{\numDim-1}} 
    \dots 
    \ketSimple{\jind{1}}
    \laplaceOnehScaled \ketSimple{\jind{0}}
    \\
    &
    + \tfrac{1}{\sqrt{\numAncillaKDof}}
    \ket{1} \ket{0}
    \ketSimple{\jind{\numDim-1}}
    \dots 
    \laplaceOnehScaled \ketSimple{\jind{1}}
    \ketSimple{\jind{0}}
    \\
    &
    + \dots
    \\
    &
    + 
    \tfrac{1}{\sqrt{\numAncillaKDof}}
    \ket{\numDim-1} \ket{0}
    \laplaceOnehScaled \ketSimple{\jind{\numDim-1}} 
    \dots 
    \ketSimple{\jind{1}}
    \ketSimple{\jind{0}}
    \\
    &
    +
    \tfrac{1}{\sqrt{\numAncillaKDof}}
    \sum_{k=0}^{\numDim-1}
    \sum_{\ell=1}^3
    \ket{k} \ket{\ell} (\dots)
    \\
    &
    +
    \tfrac{1}{\sqrt{\numAncillaKDof}}
    \sum_{k=\numDim}^{\numAncillaKDof-1}
    \ket{k} \ket{3} \ket{j}
    \\
    \to&
    \tfrac{1}{\sqrt{\numAncillaKDof}}
    \ket{0} \ket{0} 
    \Bigl(
    \ketSimple{\jind{\numDim-1}}
        \dots 
        \ketSimple{\jind{1}}
        \laplaceOnehScaled \ketSimple{\jind{0}}
        % +
        % \ketSimple{\jind{\numDim-1}}
        % \dots 
        % \laplaceOnehScaled \ketSimple{\jind{1}}
        % \ketSimple{\jind{0}}
        +
        \dots
        +
        \laplaceOnehScaled \ketSimple{\jind{\numDim-1}}
        \dots 
        \ketSimple{\jind{1}}
        \ketSimple{\jind{0}}
    \Bigr)
    + \ket{\perp}
    % \\
    % & \quad + \ket{\perp}
    \\
    &=
    \tfrac{\numDim}{\numAncillaKDof} 
    \ket{0} \ket{0} \laplaceDhScaled \ket{j}
    + \ket{\perp} \ .
\end{align*}
Here, part 2 follows by the same calculation as for Theorem~\ref{thm:block-encoding-laplace-1d} and part 3 by the identity \eqref{eq:laplaceDhScaled} and the definition of $\scalingFactor$.
\end{proof}
\begin{figure}
    \centering
    \includegraphics{
        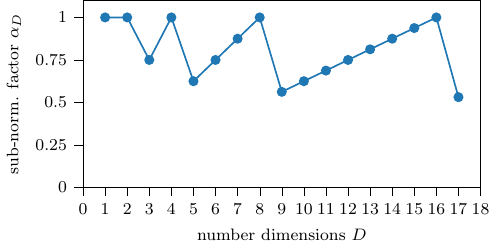}
    \caption{Number of dimensions $\numDim$ versus sub-normalization factor $\scalingFactor$. In particular, we have $\scalingFactor=1$ for $\numDim = 1$ and $2$ and $\scalingFactor=3/4$ for $\numDim=3$.}
    \label{fig:scaling-factor}
\end{figure}
%
% \noindent
In Figure~\ref{fig:scaling-factor} we show the value of the scaling factor $\scalingFactor$ versus the dimension $\numDim$.
\begin{figure}
    \centering
    \begin{subfigure}{0.4\linewidth}
        \includegraphics[width=\heatmapPlotWidth]{
        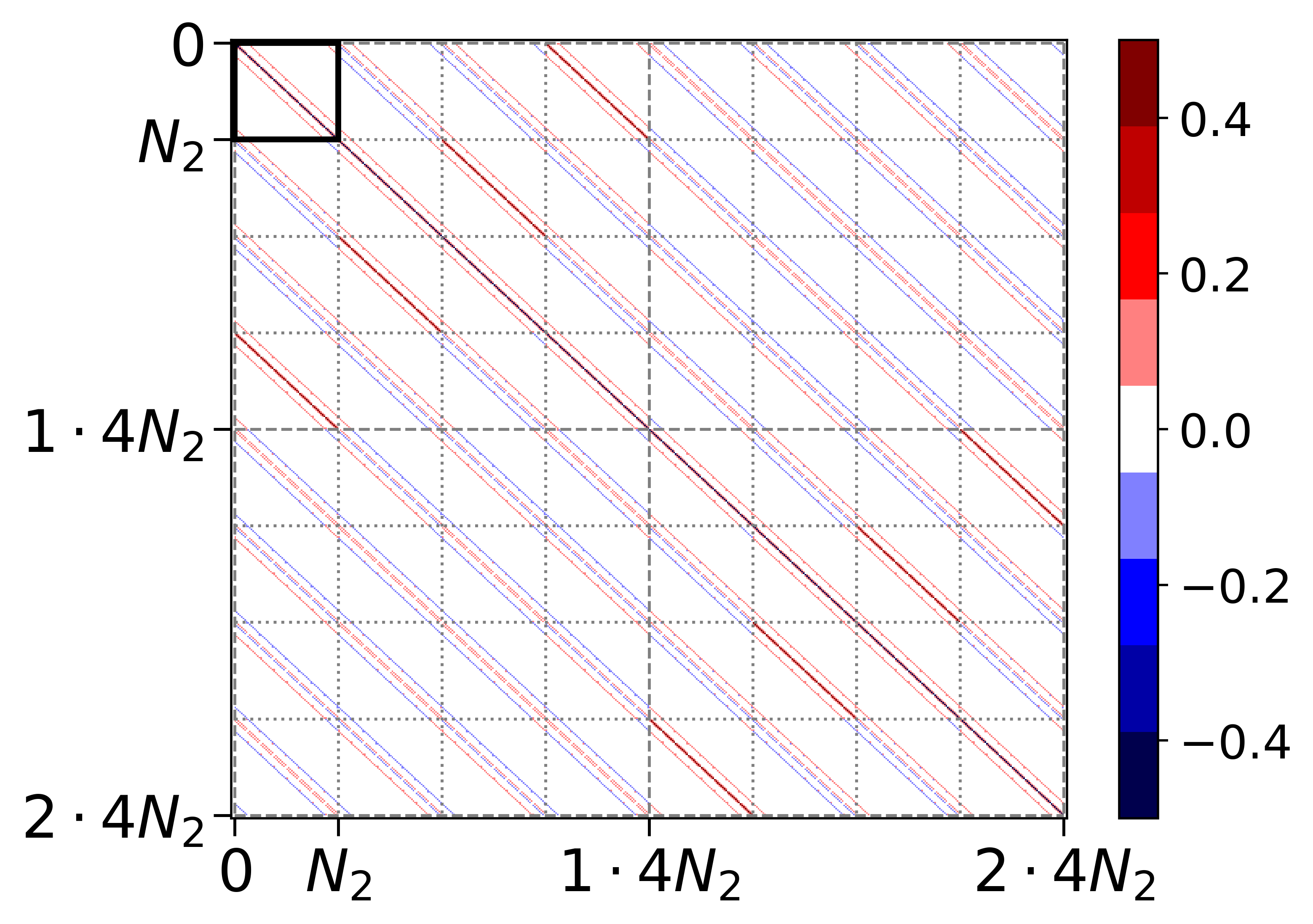}
        \caption{Entire matrix $U$.}
    \end{subfigure}
    \begin{subfigure}{0.45\linewidth}
    \includegraphics[width=\heatmapPlotWidth]{
        %figures/laplace_2d_block_encoding_zoom.pdf}
        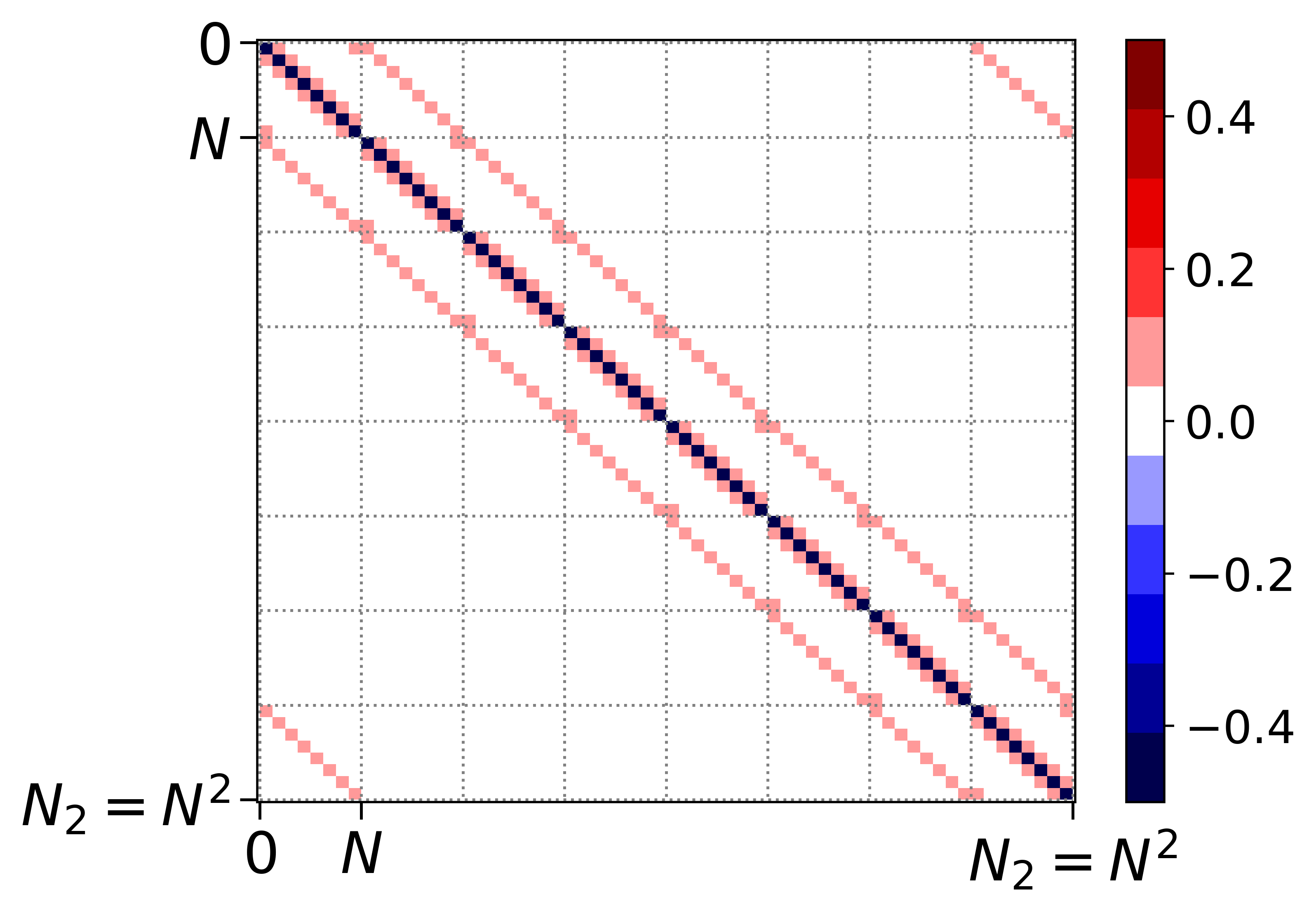}
        \caption{(1,1)-block of $U$. This corresponds to $\laplaceTwohScaled$.}
    \end{subfigure}
    \caption{Heat map plot of the entries of the matrix $U$ in \eqref{eq:circuit-laplace-arbitrary-dim-2} in Theorem~\ref{thm:block-encoding-laplace-arbitrary-dim} for $\numDim=2$ and $\numSytemBits=3$. This yields $\numAncillaKDof=2$, $\scalingFactor = 1$, and 
    % $\numDof = \numDofOne^2 = (2^3)^2 = 64$.
    $\numDofQC_2 = \numDofOne^2 = (2^3)^2 = 64$.
    The entire matrix is displayed in (a), where $\laplaceTwohScaled$ is encoded in the (1,1)-block marked with bold lines. A zoom into the (1,1)-block is given in (b).}
    \label{fig:heatmap-plot-laplace-2d}
\end{figure}
Heatmap plots of the entries of the unitary $U$ of \eqref{eq:circuit-laplace-arbitrary-dim-2} for $\numDim=2$ and $\numDim=3$ are given in Figure~\ref{fig:heatmap-plot-laplace-2d} and \ref{fig:heatmap-plot-laplace-3d}, respectively.
\begin{figure}
    \centering
    \begin{subfigure}{0.4\linewidth}
        \includegraphics[width=\heatmapPlotWidth]{
        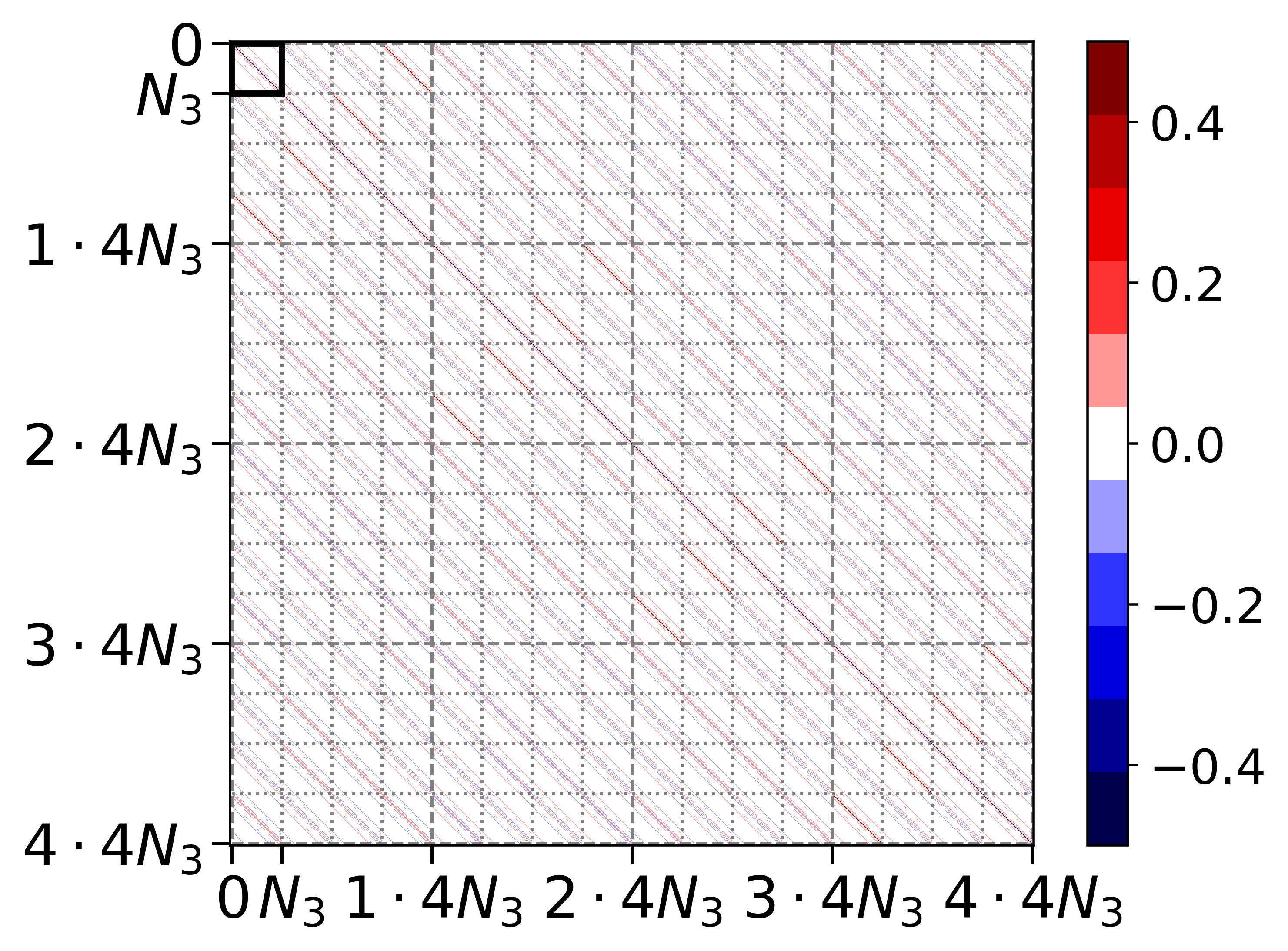}  % use pdf for high quality
        \caption{Entire matrix $U$.}
    \end{subfigure}
    \begin{subfigure}{0.45\linewidth}
        \includegraphics[width=\heatmapPlotWidth]{
        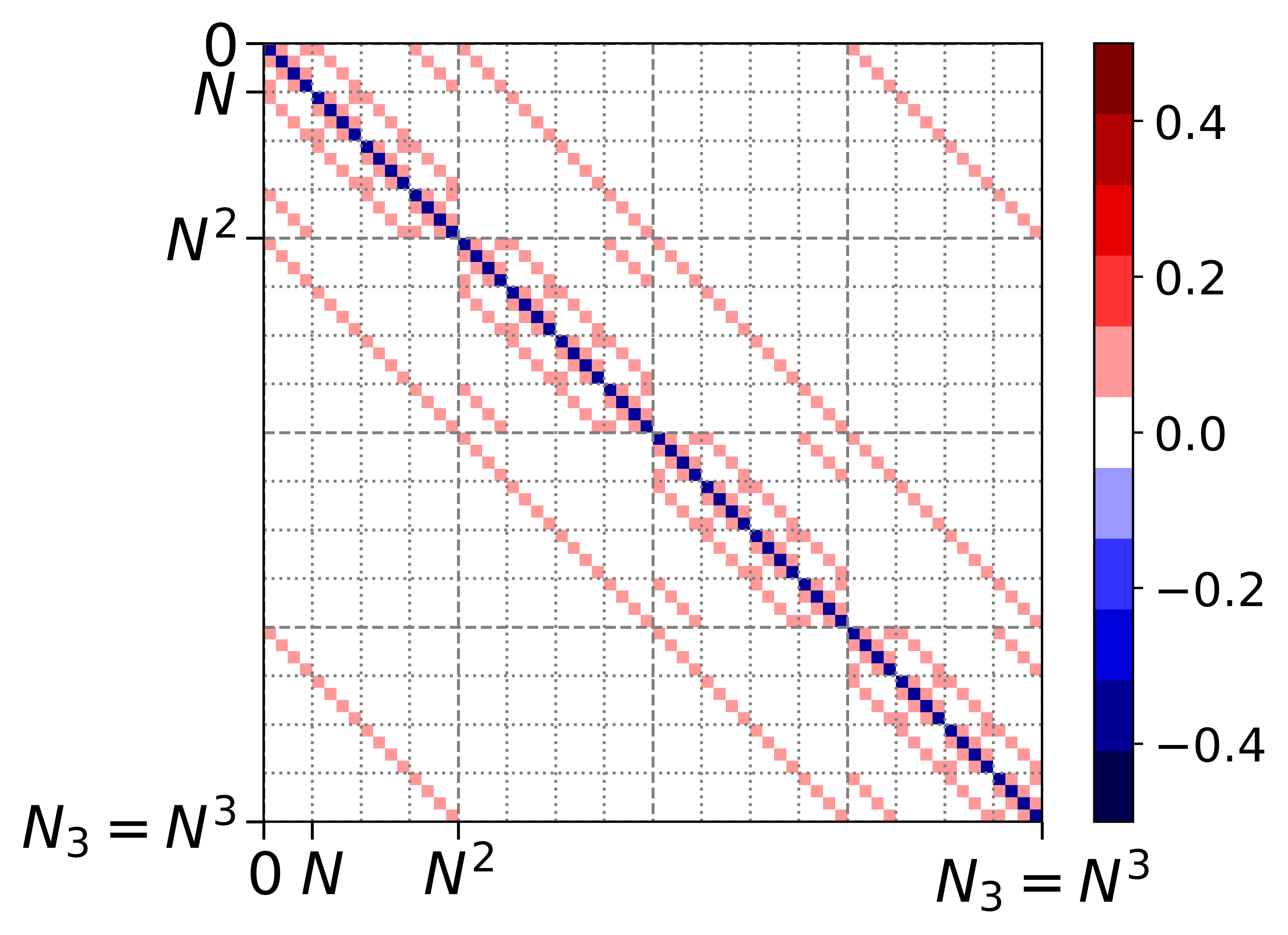}
        \caption{(1,1)-block of $U$. This corresponds to $\tfrac34 \laplaceThreehScaled$.}
    \end{subfigure}
    \caption{Heat map plot of the entries of the matrix $U$ in \eqref{eq:circuit-laplace-arbitrary-dim-2} in Theorem~\ref{thm:block-encoding-laplace-arbitrary-dim} for $\numDim=3$ and $\numSytemBits=2$. This yields $\numAncillaKDof=4$, $\scalingFactor = 3/4$, and 
    % $\numDof = \numDofOne^3 = (2^2)^3 = 64$. 
    $\numDofQC_3 = \numDofOne^3 = (2^2)^3 = 64$. 
    The entire matrix is displayed in (a), where $\tfrac34 \laplaceThreehScaled$ is encoded in the (1,1)-block marked with bold lines. A zoom into the (1,1)-block is given in (b).}
    \label{fig:heatmap-plot-laplace-3d}
\end{figure}

\section{Numerical Experiments}
\label{sec:numerical-experiments}
In this section we show numerical experiments to illustrate the scaling of the success probability and of the resource requirements of the block encoding.
\subsection{One Dimensional Laplacian}
\label{sec:numerical-experiments-one-dim}
Consider the functions
\begin{equation}
    v_1(x)
    = \sin (2 \pi x)
    \qquad
    \text{and}
    \qquad
    v_2(x) = \cos(6 \pi x) \ .
    \label{eq:num-exp-laplace-1d-1}
\end{equation}
We have $\laplaceOne v_1 = - (2 \pi)^2 v_1$ and $\laplaceOne v_2 = - (6 \pi)^2 v_2$. By \eqref{eq:theorem-encoding-laplace-1d-3} the success probability of the block encoding \eqref{eq:theorem-encoding-laplace-1d-1} scales with $C_1 h^4$ for $v_1$ and $C_2 h^4$ for $v_2$, where
\begin{equation}
    C_1
    = \frac{(2 \pi)^4}{16}
    = \pi^4 \ ,
    \qquad
    C_2
    = \frac{(6 \pi)^4}{16}
    = 81 \pi^4 \ .
    \label{eq:num-exp-laplace-1d-2}
\end{equation}
For the block encoding \eqref{eq:camps-circuit-laplace-1d} proposed in \cite{Camps.2022} the constants are a factor $16$ smaller, as we detailed in Section~\ref{sec:other-block-encodings}. This behavior is confirmed in Figure~\ref{fig:num-exp-laplace-1d-success-prob}, where we plot the success probability for both block encodings for $v_1$ and $v_2$ for different grid widths $h$. In particular, we see the factor $16$ improvement of the block encoding presented in this paper and the decrease with $h^4$ as $h \to 0$.
\begin{figure}
    \centering
    \begin{subfigure}{0.95\textwidth}
    \centering
        \includegraphics{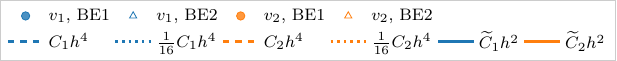}
        \caption*{}
    \end{subfigure}
    \begin{subfigure}{0.4\textwidth}
        \includegraphics{
            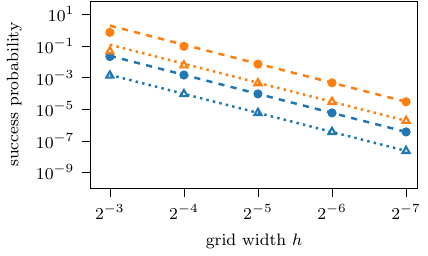}
            %figures/laplace_1d_success_probability.pdf}
        \caption{Success probability of block encoding.}
        \label{fig:num-exp-laplace-1d-success-prob}
    \end{subfigure}
    \qquad \ \ \ \
    \begin{subfigure}{0.4\textwidth}
        \includegraphics{
            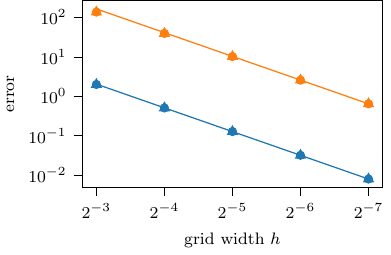}
            % figures/laplace_1d_error.pdf}
        \caption{Error due to finite difference approximation.}
        \label{fig:num-exp-laplace-1d-error}
    \end{subfigure}
    \caption{Block encoding of the one dimensional, discretized Laplacian $\laplaceOneh$. We compare the encoding presented in this paper (labeled as BE1 and displayed with circle markers in the plots) and the one of \cite{Camps.2022} (BE2, triangle markers), see also \eqref{eq:theorem-encoding-laplace-1d} and \eqref{eq:camps-circuit-laplace-1d}, respectively. (a) Success probability \eqref{eq:definition-block-encoding-4} of BE1 and BE2 for functions $v_1$ (blue markers) and $v_2$ (orange markers). (b) Error \eqref{eq:max-error} due to the finite difference discretization. The functions $v_i$ are defined in  \eqref{eq:num-exp-laplace-1d-1} and the constants $C_i$ and $\widetilde{C}_i$ are given in \eqref{eq:num-exp-laplace-1d-2}.}
    \label{fig:num-exp-laplace-1d}
\end{figure}
Moreover, we see in Figure~\ref{fig:num-exp-laplace-1d-success-prob} that we have a higher success probability for the higher oscillatory function $v_2$ compared to $v_1$. However, this must be contrasted with the fact that the error
\begin{equation}
    e_\mathrm{max}(v)
    = \Bigl\|
    \lambda_{1, \mathrm{max}} 
    \norm{\vecv}_2 
    \tfrac{1}{\scalingFactorBE }\laplaceOnehScaled \ket{v} 
    - \laplaceOne \vecv
    \Bigr\|_{\infty}
    \label{eq:max-error}
\end{equation}
of the finite difference approximation scales with the fourth derivative $v^{(4)}$, 
\begin{equation*}
    e_\mathrm{max}(v) 
    \le \tfrac{h^2}{12} \max_{x \in [0, 1]} |v^{(4)}(x)| \ ,
\end{equation*}
see \eqref{eq:taylor-expansion}.
In our case, this means $e_\mathrm{max}(v_i) \le \widetilde{C}_i h^2$ with
\begin{equation}
    \widetilde{C}_1
    = \frac{(2\pi)^4}{12}
    = \frac{4\pi^4}{3} \ ,
    \qquad
    \widetilde{C}_2
    = \frac{(6\pi)^4}{12}
    = 108 \pi^4 \ .
    \label{eq:max-error-constants}
\end{equation}
So, to obtain the same finite difference error, the function $v_2$ needs a much finer grid than $v_1$, see Figure~\ref{fig:num-exp-laplace-1d-error}, which decreases the success probability.
\subsection{Multi-Dimensional Laplacian}
\label{sec:numerical-experiments-multi-dim}
For dimensions $\numDim=1, 2, 3,$ and $4$ let us consider the functions
\begin{equation}
    v_{\numDim}(\vecx) 
    = \prod_{d=0}^{\numDim-1}
    \sin( 2 \pi \vecxentry{d}) \ ,
    \label{eq:num-exp-2-laplace-multi-d-1}
\end{equation}
which satisfy $\laplaceD v_\numDim = (-1) \numDim (2 \pi)^2 v_\numDim$.
By \eqref{eq:circuit-laplace-power-2d-4}, the success probability of the block encoding \eqref{eq:circuit-laplace-arbitrary-dim-1} scales with $C_\numDim h^4$, where
\begin{equation*}
    C_\numDim
    = \frac{\numDim^2 (2\pi)^4}{16 \numAncillaKDof^2}
    = \frac{\pi^4 \numDim^2}{\numAncillaKDof^2} \ .
\end{equation*}
Thus, we have
\begin{equation}
    C_1 = C_2 = C_4 = \pi^4 \, ,
    \quad
    C_3 = \tfrac{9}{16} \pi^4 \, .
    % \qquad
    % \text{where }
    % C = \pi^4 \ .
    \label{eq:num-exp-2-laplace-multi-d-2}
\end{equation}
In Figure~\ref{fig:num-exp-laplace-multi-d-experiment-2} we show the success probabilities for the functions $v_\numDim$ for different grid widths $h$. We visualize the dependence on $h$ by showing the number of grid points $\numDof = \numDofOne^\numDim = 1/h^\numDim$ on the $x$-axis.
\begin{figure}
    \centering
    \includegraphics{
        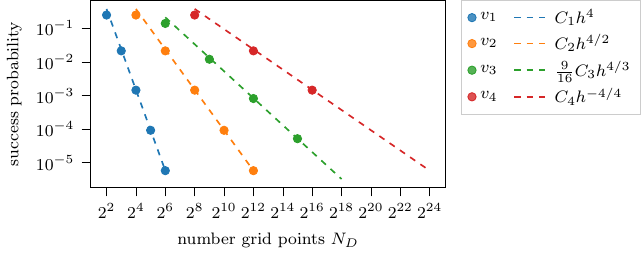}
    \caption{Number of grid points $\numDof = \numDofOne^\numDim = 2^{\numSytemBits \numDim}$ versus success probability \eqref{eq:definition-block-encoding-4} of encoding the $\numDim$ dimensional, discretized Laplacian $\laplaceDh$ with the block encoding \eqref{eq:circuit-laplace-arbitrary-dim-1} for dimensions $\numDim=1, 2, 3,$ and $4$ (blue, orange, green, and red markers). The functions $v_\numDim$ are defined in \eqref{eq:num-exp-2-laplace-multi-d-1} and the constants $C_\numDim$ are given in \eqref{eq:num-exp-2-laplace-multi-d-2}.}
    \label{fig:num-exp-laplace-multi-d-experiment-2}
\end{figure}
\subsection{Resource Estimates}
\label{sec:resource-estimates}
Because of the unavoidable presence of noise in quantum computers, a reliable execution of a quantum circuit is only possible if it is endowed with a quantum error correction protocol \cite{Nielsen.2012,Terhal.2015}. However, this protection comes with an overhead in quantum resources that needs to be taken into account when assessing the complexity of a quantum algorithm. One central aspect required to calculate this overhead is the number of $T$-gates that a quantum circuit contains when it is expressed in the universal gate set Clifford+$T$. In order to estimate the required number of $T$ gates, we implemented the block encoding circuits \eqref{eq:theorem-encoding-laplace-1d}, \eqref{eq:camps-circuit-laplace-1d} and \eqref{eq:circuit-laplace-arbitrary-dim} in \texttt{qualtran} \cite{Harrigan.2024}. The shift operators $\ShiftPlusMinus$ are implemented with the \texttt{AddK} class of \texttt{qualtran}. At the expense of additional auxiliary qubits, this implementation of the addition has the lowest known $T$ gate count and a logarithmic scaling $\mathcal{O}(\log \numDof)$ in the number of grid points  \cite{Draper.2006,Haner.2020}.

In Figure~\ref{fig:resource-estimates} we show the number of $T$ gates to encode $\laplaceDh$ for $\numDim = 1, 2,$ and $3$ with the block encoding \eqref{eq:circuit-laplace-arbitrary-dim}. We can clearly see the logarithmic scaling of the number of $T$ gates with the number of grid points $\numDof$.
In Figure~\ref{fig:resource-estimates-compare-camps} we show a comparison of the block encoding \eqref{eq:theorem-encoding-laplace-1d} presented in this paper and the block encoding \eqref{eq:camps-circuit-laplace-1d} from \cite{Camps.2022} to encode $\laplaceOneh$. Besides the higher success probability of the block encoding \eqref{eq:theorem-encoding-laplace-1d}, that we discussed in Section~\ref{sec:numerical-experiments-one-dim}, we also observe a lower $T$ gate complexity for this block encoding. The reason for this is that the circuit of \eqref{eq:camps-circuit-laplace-1d} has additional $RY$ gates that are not present in \eqref{eq:theorem-encoding-laplace-1d}.
\begin{figure}
    \centering
    \begin{subfigure}{0.95\textwidth}
        \centering
        \includegraphics{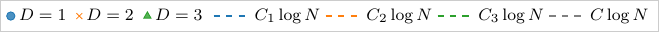}
        \caption*{}
    \end{subfigure}
    \begin{subfigure}[b]{0.4\linewidth}
        \centering
        \includegraphics{
            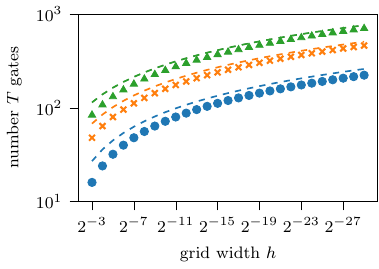}
        % \\[0.25cm]
        \caption{Number of $T$ gates versus grid width $h$.\\ \phantom{a}}
    \end{subfigure}
    \qquad
    \begin{subfigure}[b]{0.43\linewidth}
        \centering
        \includegraphics{
            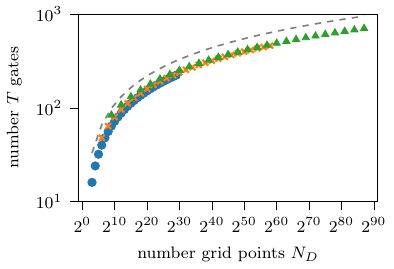}
        \caption{Number of $T$ gates versus number grid points $\numDof$.}
    \end{subfigure}
    \caption{Number of $T$ gates versus grid width $h$ (a) and versus number grid points $\numDof$ (b) for block encoding $\laplaceDh$ with the block encoding \eqref{eq:circuit-laplace-arbitrary-dim}. Note that we have  $\numDof = \numDofOne^\numDim = 1/h^\numDim$.  The constants are $C_1=9$, $C_2=17$, $C_3=25$, and $C=11$.}
    \label{fig:resource-estimates}
\end{figure}
\begin{figure}
    \centering
    \includegraphics{
        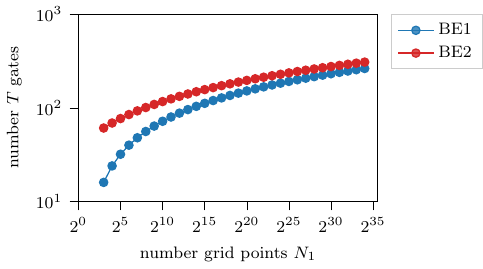}
        % figures/resource_estimates_camps.pdf}
    \caption{Number of $T$ gates versus number of grid points $\numDofQC_1 = \numDofOne = 1/h$ for encoding $\laplaceOneh$ with the block encoding \eqref{eq:theorem-encoding-laplace-1d} (blue markers) and  \eqref{eq:camps-circuit-laplace-1d} (red markers).}
    \label{fig:resource-estimates-compare-camps}
\end{figure}

\section{Conclusion}
In this paper, we have developed an efficient method for the block encoding of finite difference discretizations of the Laplacian operator with periodic boundary conditions. Our approach leverages the unique structure of finite difference methods, allowing us to construct explicit quantum circuits that improve existing encodings. We have analytically derived the scaling of the sub-normalization factor and the success probability of the block encoding in relation to the problem dimension, the grid width, and the regularity of the exact solution.
The results demonstrate that our method not only enhances the sub-normalization factor and consequently the success probability, but also reduces the resource requirements, making it a practical solution for quantum algorithms aiming for the solution of differential equations involving the Laplacian. The implications of our findings extend beyond the discrete Laplacian, as the techniques and insights gleaned from this research can be adapted to other differential operators. Future work will aim to incorporate other boundary conditions, as for example Dirichlet boundary conditions \cite{Camps.2022,Sunderhauf.2024,Kharazi.2025}, and to compare the block encoding presented here with approximative \cite{Camps_2022} and variational methods \cite{Kikuchi.2023,rullkotter2025}.

\section*{Acknowledgment}
The authors thank Leon Rullk\"otter for carefully proofreading this manuscript.
\\
This work was funded by the Ministry of Economic Affairs, Labour, and Tourism Baden-W\"urttemberg in the frame of the Competence Center Quantum Computing Ba\-den-W\"urt\-tem\-berg.

\appendix

\section{Block Encoding of First Order Differential Operators}
\label{sec:block-encoding-first-order}
The construction of the block encoding of the finite difference approximation of the Laplacian $\laplaceDhScaled$ can be transferred to other differential operators. In this appendix we show this for first order differential operators. For example, the circuit
\begin{subequations}
\begin{equation}
    \begin{array}{c}
    \Qcircuit @C=1em @R=.7em {
    \lstick{\ket{j}} 
    & {/} \qw 
    & \barrier[-1.5em]{1} \qw 
    & \gate{\ShiftMinus} 
    & \gate{\ShiftPlus} \barrier[-1.2em]{1} 
    & \qw 
    & \qw 
    % &
    \\
    \lstick{\ket{\ell}} 
    & \gate{H} 
    & \gate{Z} 
    & \ctrlo{-1}
    & \ctrl{-1}
    & \gate{H} 
    % & \meter
    & \qw
    }

    \end{array}
    \label{eq:circuit-block-encoding-dx}
\end{equation}
encodes
\begin{equation}
    \DOnehScaled
    = h \DOneh 
\end{equation}
with $\numAncillaBits = 1$ and $\scalingFactorBE = 1$.
\end{subequations}
Here, $\DOneh$ is the central finite difference approximation to $\tfrac{d}{dx}$ defined in \eqref{eq:doneh}.
From \eqref{eq:circuit-block-encoding-dx} we can derive
\begin{subequations}
\begin{equation}
    \begin{array}{c}
     \Qcircuit @C=1em @R=.6em {
    \lstick{\ket{\jind{0}}} 
    & {/} \qw \barrier[-1.2em]{3}
    & \qw 
    & \qw 
    & \gate{\ShiftMinus} 
    & \gate{\ShiftPlus} 
    & \qw 
    & \qw 
    & \qw \barrier[-.5em]{3}
    & \qw
    & \qw 
    \\
    \lstick{\ket{\jind{1}}} 
    & {/} \qw 
    & \qw 
    & \qw 
    & \qw 
    & \qw 
    & \gate{\ShiftMinus} 
    & \gate{\ShiftPlus} 
    & \qw
    & \qw 
    & \qw
    \\
    \lstick{\ket{k}} 
    & \gate{H} 
    & \qw 
    & \qw
    & \ctrlo{-2} 
    & \ctrlo{-2}
    & \ctrl{-1} 
    & \ctrl{-1}
    & \qw 
    & \qw
    & \qw
    \\
    \lstick{\ket{\ell}} 
    & \qw 
    & \gate{H} 
    & \gate{Z}
    & \ctrlo{-1} 
    & \ctrl{-1}
    & \ctrlo{-1} 
    & \ctrl{-1}
    & \gate{H}
    & \qw 
    & \qw
}
    \end{array}
\end{equation}
which encodes a finite difference approximation of the gradient $\mathrm{grad} = (\partial/\partial x^{(0)}, \partial / \partial x^{(1)})^T$ with $\numAncillaBits=2$ and $\scalingFactorBE=1/\sqrt{2}$ as
\begin{equation}
    U
    = \pmat{
    \tfrac{1}{\sqrt{2}} \DOnehScaledx & \ast &
    \\
    \tfrac{1}{\sqrt{2}} \DOnehScaledy & \ast & 
    \\
    \ast & \ast
    } \ ,
\end{equation}
where
\begin{equation}
    \DOnehScaledx
    = I \otimes \DOnehScaled \ 
    \qquad \text{and} \qquad
    \DOnehScaledy
    = \DOnehScaled \otimes I \ .
\end{equation}
This means
\begin{equation}
    U \ket{0}\ket{0}\ket{j}
    =
    \tfrac{1}{\sqrt{2}}\ket{0} \Bigl(
    \ket{0} \DOnehScaledx \ket{j}
    +
    \ket{1} \DOnehScaledy \ket{j}
    \Bigr)
    +  \ket{\perp}
    \ .
\end{equation}
\end{subequations}
Moreover, the circuit
\begin{subequations}
\begin{equation}
    \begin{array}{c}
     \Qcircuit @C=1em @R=.6em {
    \lstick{\ket{\jind{0}}} 
    & {/ \ \ } \qw \barrier[-1.2em]{3}
    & \qw 
    & \qw 
    & \gate{\ShiftMinus} 
    & \gate{\ShiftPlus} 
    & \qw 
    & \qw 
    & \qw \barrier[-1.2em]{3}
    & \qw
    & \qw 
    \\
    \lstick{\ket{\jind{1}}} 
    & {/ \ \ } \qw 
    & \qw 
    & \qw 
    & \qw 
    & \qw 
    & \gate{\ShiftMinus} 
    & \gate{\ShiftPlus} 
    & \qw
    & \qw 
    & \qw
    \\
    \lstick{\ket{k}} 
    & \qw
    & \qw 
    & \qw
    & \ctrlo{-2} 
    & \ctrlo{-2}
    & \ctrl{-1} 
    & \ctrl{-1}
    & \qw 
    & \gate{H} 
    & \qw
    \\
    \lstick{\ket{\ell}} 
    & \qw 
    & \gate{H} 
    & \gate{Z}
    & \ctrlo{-1} 
    & \ctrl{-1}
    & \ctrlo{-1} 
    & \ctrl{-1}
    & \gate{H}
    & \qw 
    & \qw
}
    \end{array}
\end{equation}
encodes the divergence $\mathrm{div} = (\partial/\partial x^{(0)}, \partial / \partial x^{(1)})$ with $\numAncillaBits=2$ and $\scalingFactorBE=1/\sqrt{2}$ as
\begin{equation}
    U
    = \pmat{
    \tfrac{1}{\sqrt{2}} \DOnehScaledx 
    & \tfrac{1}{\sqrt{2}} \DOnehScaledy 
    & \ast 
    &
    \\
    \ast & \ast & \ast
    \\
    \ast & \ast & \ast
    } \ .
\end{equation}
\end{subequations}
Finally, the circuit
\begin{subequations}
\begin{equation}
    \begin{array}{c}
     \Qcircuit @C=1em @R=.6em {
    \lstick{\ket{\jind{0}}} 
    & {/} \qw \barrier[-1.2em]{4}
    & \qw 
    & \qw 
    & \gate{\ShiftMinus} 
    & \gate{\ShiftPlus} 
    & \qw 
    & \qw 
    & \qw \barrier[-1.2em]{4}
    & \qw
    & \qw 
    \\
    \lstick{\ket{\jind{1}}} 
    & {/} \qw 
    & \qw 
    & \qw 
    & \qw 
    & \qw 
    & \gate{\ShiftMinus} 
    & \gate{\ShiftPlus} 
    & \qw
    & \qw 
    & \qw
    \\
    \lstick{\ket{k_1}} 
    & \gate{H}
    & \qw 
    & \qw
    & \ctrlo{-2} 
    & \ctrlo{-2}
    & \ctrl{-1} 
    & \ctrl{-1}
    & \qw 
    & \gate{H} 
    & \qw
    \\
    \lstick{\ket{k_0}} 
    & \ctrl{-1}
    & \qw 
    & \qw
    & \qw
    & \qw
    & \qw
    & \qw
    & \qw 
    & \ctrlo{-1} 
    & \qw
    \\
    \lstick{\ket{\ell}} 
    & \qw 
    & \gate{H} 
    & \gate{Z}
    & \ctrlo{-2} 
    & \ctrl{-2}
    & \ctrlo{-2} 
    & \ctrl{-2}
    & \gate{H}
    & \qw 
    & \qw
}
    \end{array}
\end{equation}
encodes the first order wave equation operator 
\begin{equation}
    \mathrm{W} 
    = \pmat{
    0 & 0 & \partial/\partial x^{(0)}
    \\
    0 & 0 & \partial / \partial x^{(1)}
    \\
    \partial/\partial x^{(0)} & \partial / \partial x^{(1)} & 0
    }
\end{equation}
with $\numAncillaBits=2$ and $\scalingFactorBE=1/\sqrt{2}$ as
\begin{equation}
    U
    = \pmat{
    0 
    & 0
    & \tfrac{1}{\sqrt{2}} \DOnehScaledx 
    & \ast
    \\
    0
    & 0
    & \tfrac{1}{\sqrt{2}} \DOnehScaledy 
    & \ast 
    \\
    \tfrac{1}{\sqrt{2}} \DOnehScaledx 
    & \tfrac{1}{\sqrt{2}} \DOnehScaledy
    & 0
    & \ast
    \\
    \ast & \ast & \ast & \ast
    } \ .
\end{equation}
\end{subequations}

\bibliography{literature}

\begin{thebibliography}{10}

\bibitem{Evans.2022}
L.~C. Evans.
\newblock ``Partial differential equations''.
\newblock Volume~19 of Graduate studies in mathematics.
\newblock {American Mathematical Society}. Providence, Rhode Island~(2022).
\newblock Second edition.

\bibitem{Smith.1993}
G.~D. Smith.
\newblock ``Numerical solution of partial differential equations: Finite
  difference methods''.
\newblock Oxford applied mathematics and computing science series. {Clarendon
  Press}. Oxford~(1993).
\newblock 3. ed., repr. (with corr.) edition.

\bibitem{Thomas.1995}
J.~W. Thomas.
\newblock ``Numerical partial differential equations: Finite difference
  methods''.
\newblock \href{https://dx.doi.org/10.1007/978-1-4899-7278-1}{Volume~22 of
  Springer eBook Collection Mathematics and Statistics}.
\newblock Springer. New York, NY~(1995).

\bibitem{Hairer.1991}
E.~Hairer and G.~Wanner.
\newblock ``Solving ordinary differential equations {II}''.
\newblock \href{https://dx.doi.org/10.1007/978-3-662-09947-6}{Volume~14}.
\newblock {Springer Berlin Heidelberg}. Berlin, Heidelberg~(1991).

\bibitem{Sogabe.2022}
T.~Sogabe.
\newblock ``Krylov subspace methods for linear systems''.
\newblock \href{https://dx.doi.org/10.1007/978-981-19-8532-4}{Volume~60}.
\newblock {Springer Nature Singapore}. Singapore~(2022).

\bibitem{Low.2019}
G.~H. Low and I.~L. Chuang.
\newblock ``Hamiltonian simulation by qubitization''.
\newblock \href{https://dx.doi.org/10.22331/q-2019-07-12-163}{Quantum {\bf 3},
  163}~(2019).

\bibitem{Gilyen.2019}
A.~Gily{\'e}n, Y.~Su, G.~H. Low, and N.~Wiebe.
\newblock ``Quantum singular value transformation and beyond: exponential
  improvements for quantum matrix arithmetics''.
\newblock In Proceedings of the 51st Annual ACM SIGACT Symposium on Theory of
  Computing.
\newblock \href{https://dx.doi.org/10.1145/3313276.3316366}{Pages 193--204}.
\newblock ACM Digital Library. {Association for Computing Machinery}~(2019).

\bibitem{Berry.2007}
D.~W. Berry, G.~Ahokas, R.~Cleve, and B.~C. Sanders.
\newblock ``Efficient quantum algorithms for simulating sparse
  {H}amiltonians''.
\newblock \href{https://dx.doi.org/10.1007/s00220-006-0150-x}{Communications in
  Mathematical Physics {\bf 270}, 359--371}~(2007).

\bibitem{Berry.2015}
D.~W. Berry, A.~M. Childs, and R.~Kothari.
\newblock ``Hamiltonian simulation with nearly optimal dependence on all
  parameters''.
\newblock In 2015 IEEE 56th Annual Symposium on Foundations of Computer Science
  (FOCS 2015).
\newblock \href{https://dx.doi.org/10.1109/FOCS.2015.54}{Pages 792--809}.
\newblock Piscataway, NJ~(2015). IEEE.

\bibitem{Childs.2017}
A.~M. Childs, R.~Kothari, and R.~D. Somma.
\newblock ``Quantum algorithm for systems of linear equations with
  exponentially improved dependence on precision''.
\newblock \href{https://dx.doi.org/10.1137/16M1087072}{SIAM Journal on
  Computing {\bf 46}, 1920--1950}~(2017).

\bibitem{Camps.2022}
D.~Camps, L.~Lin, R.~Van~Beeumen, and C.~Yang.
\newblock ``Explicit quantum circuits for block encodings of certain sparse
  matrices''.
\newblock \href{https://dx.doi.org/10.1137/22M1484298}{SIAM Journal on Matrix
  Analysis and Applications {\bf 45}, 801--827}~(2024).

\bibitem{Kharazi.2025}
T.~Kharazi, A.~M. Alkadri, J.-P. Liu, K.~K. Mandadapu, and K.~B. Whaley.
\newblock ``Explicit block encodings of boundary value problems for many-body
  elliptic operators''.
\newblock \href{https://dx.doi.org/10.22331/q-2025-06-04-1764}{Quantum {\bf 9},
  1764}~(2025).

\bibitem{Low.2017}
G.~H. Low and I.~L. Chuang.
\newblock ``Optimal {H}amiltonian simulation by quantum signal processing''.
\newblock \href{https://dx.doi.org/10.1103/PhysRevLett.118.010501}{Physical
  review letters {\bf 118}, 010501}~(2017).

\bibitem{Martyn.2021}
J.~M. Martyn, Z.~M. Rossi, A.~K. Tan, and I.~L. Chuang.
\newblock ``Grand unification of quantum algorithms''.
\newblock \href{https://dx.doi.org/10.1103/PRXQuantum.2.040203}{PRX Quantum
  {\bf 2}, 040203}~(2021).

\bibitem{Nielsen.2012}
M.~A. Nielsen and I.~L. Chuang.
\newblock ``Quantum computation and quantum information''.
\newblock \href{https://dx.doi.org/10.1017/CBO9780511976667}{{Cambridge
  University Press}}. ~(2012).

\bibitem{Rieffel.2014}
E.~Rieffel and W.~Polak.
\newblock ``Quantum computing: A gentle introduction''.
\newblock Scientific and engineering computation. {The {MIT} Press}. Cambridge,
  Massachusetts and London, England~(2014).
\newblock First {MIT} press paperback edition.

\bibitem{Lin.2022}
L.~Lin.
\newblock ``Lecture notes on quantum algorithms for scientific
  computation''~(2022).
\newblock  \href{http://arxiv.org/abs/2201.0830}{arXiv:2201.0830}.

\bibitem{Terhal.2015}
B.~M. Terhal.
\newblock ``Quantum error correction for quantum memories''.
\newblock \href{https://dx.doi.org/10.1103/RevModPhys.87.307}{Reviews of Modern
  Physics {\bf 87}, 307--346}~(2015).

\bibitem{Harrigan.2024}
M.~P. Harrigan, T.~Khattar, C.~Yuan, A.~Peduri, N.~Yosri, F.~D. Malone,
  R.~Babbush, and N.~C. Rubin.
\newblock ``Expressing and analyzing quantum algorithms with qualtran''~(2024)
  \href{http://arxiv.org/abs/2409.04643}{arXiv:2409.04643}.

\bibitem{Draper.2006}
T.~G. Draper, S.~A. Kutin, E.~M. Rains, and K.~M. Svore.
\newblock ``A logarithmic-depth quantum carry-lookahead adder''.
\newblock Quantum Info. Comput. {\bf 6}, 351–369~(2006).
\newblock  url:~\url{https://dl.acm.org/doi/abs/10.5555/2012086.2012090}.

\bibitem{Haner.2020}
T.~H{\"a}ner, S.~Jaques, M.~Naehrig, M.~Roetteler, and M.~Soeken.
\newblock ``Improved quantum circuits for elliptic curve discrete logarithms''.
\newblock In Jintai Ding and Jean-Pierre Tillich, editors, Post-Quantum
  Cryptography.
\newblock \href{https://dx.doi.org/10.1007/978-3-030-44223-1{\textunderscore
  }23}{Volume 12100 of Lecture Notes in Computer Science, pages 425--444}.
\newblock {Springer International Publishing}~(2020).

\bibitem{Sunderhauf.2024}
C.~S{\"u}nderhauf, E.~Campbell, and J.~Camps.
\newblock ``Block-encoding structured matrices for data input in quantum
  computing''.
\newblock \href{https://dx.doi.org/10.22331/q-2024-01-11-1226}{Quantum {\bf 8},
  1226}~(2024).

\bibitem{Camps_2022}
D.~Camps and R.~Van~Beeumen.
\newblock ``{FABLE}: Fast approximate quantum circuits for block-encodings''.
\newblock In 2022 IEEE International Conference on Quantum Computing and
  Engineering (QCE).
\newblock \href{https://dx.doi.org/10.1109/qce53715.2022.00029}{Page
  104–113}.
\newblock IEEE~(2022).

\bibitem{Kikuchi.2023}
Y.~Kikuchi, C.~{McKeever}, L.~Coopmans, M.~Lubasch, and M.~Benedetti.
\newblock ``Realization of quantum signal processing on a noisy quantum
  computer''.
\newblock \href{https://dx.doi.org/10.1038/s41534-023-00762-0}{npj Quantum
  Information {\bf 9}, 93}~(2023).

\bibitem{rullkotter2025}
L.~Rullkötter, S.~Weber, V.~M. Katukuri, C.~Tutschku, and B.~C. Mummaneni.
\newblock ``Resource-efficient variational block-encoding''~(2025).
\newblock  \href{http://arxiv.org/abs/2507.17658}{arXiv:2507.17658}.

\end{thebibliography}
\bibliographystyle{quantum}

\end{document}